\documentclass[sigconf,authorversion,nonacm]{acmart}
\AtBeginDocument{%
  \providecommand\BibTeX{{%
    \normalfont B\kern-0.5em{\scshape i\kern-0.25em b}\kern-0.8em\TeX}}}

\usepackage[font=small,skip=0pt]{caption}
\usepackage{amsmath}
\usepackage{multirow}
\usepackage{url}
\usepackage[linesnumbered,ruled,vlined]{algorithm2e}
\usepackage{xspace}
\usepackage{graphicx}
\usepackage{bm}
\usepackage{enumitem,kantlipsum}
\usepackage{subcaption}
\usepackage{cleveref}
\usepackage{listings}
\usepackage{float}
\usepackage{xcolor}
\usepackage{titlesec}
\usepackage{setspace}
\usepackage[final]{pdfpages}
\usepackage{balance}

\makeatletter
\newcommand{\removelatexerror}{\let\@latex@error\@gobble}
\makeatother

\newcommand{\veryshortarrow}[1][3pt]{\mathrel{%
   \vcenter{\hbox{\rule[-.5\fontdimen8\textfont3]{#1}{\fontdimen8\textfont3}}}%
   \mkern-4mu\hbox{\usefont{U}{lasy}{m}{n}\symbol{41}}}}

\newcommand{\dotarrow}{\cdotp\veryshortarrow}

\theoremstyle{definition}
\newtheorem{definition}{Definition}[section]

\theoremstyle{remark}

\newtheorem{remark}{Remark}[section]
\theoremstyle{property}
\newtheorem{property}{Property}[section]

\newcommand{\stitle}[1]{\vspace{0.5ex} \noindent{{\bf #1}}}

\newcommand{\sstitle}[1]{\vspace{0.5ex} \noindent{\textit{ #1}}}

\newcommand{\kw}[1]{{\ensuremath {\mathsf{#1}}}\xspace}

\newcommand{\kwnospace}[1]{{\ensuremath {\mathsf{#1}}}}

\newcommand{\ei}{\end{itemize}}
\newcommand{\ee}{\end{enumerate}}

\newcommand{\beqn}{\begin{eqnarray*}}
\newcommand{\eeqn}{\end{eqnarray*}}

\newcounter{ccc}

\newcommand{\eat}[1]{}

\long\def\comment#1{}

\newcommand{\page}{\kw{HUGE}}
\newcommand{\pagewco}{\kwnospace{HUGE-}\kw{WCO}}
\newcommand{\pageseed}{\kwnospace{HUGE-}\kw{SEED}}

\newcommand{\pagebenu}{\kwnospace{HUGE-}\kw{BENU}}
\newcommand{\pagerads}{\kwnospace{HUGE-}\kw{RADS}}
\newcommand{\pageeh}{\kwnospace{HUGE-}\kw{EH}}
\newcommand{\pagegf}{\kwnospace{HUGE-}\kw{GF}}

\newcommand{\edgejoin}{\kwnospace{Edge}\kw{Join}}
\newcommand{\ttjoin}{\kwnospace{Twin}\kwnospace{Twig}\kw{Join}}
\newcommand{\starjoin}{\kwnospace{Star}\kw{Join}}
\newcommand{\multiwayjoin}{\kwnospace{Multiway}\kw{Join}}

\newcommand{\seed}{\kw{SEED}}
\newcommand{\bigjoin}{\kwnospace{BiG}\kw{Join}}
\newcommand{\psgl}{\kw{PSgL}\xspace}
\newcommand{\cbf}{\kw{CBF}\xspace}

\newcommand{\benu}{\kw{BENU}}
\newcommand{\rads}{\kw{RADS}}

\newcommand{\eh}{\kwnospace{Empty}\kw{Headed}}
\newcommand{\gf}{\kwnospace{Graph}\kw{Flow}}
\newcommand{\oom}{\texttt{OOM}\xspace}
\newcommand{\timeout}{\texttt{OT}\xspace}

\newcommand{\reffig}[1]{Figure~\ref{fig:#1}}
\newcommand{\refsec}[1]{Section~\ref{sec:#1}}
\newcommand{\reftable}[1]{Table~\ref{tab:#1}}
\newcommand{\refalg}[1]{Algorithm~\ref{alg:#1}}
\newcommand{\refeq}[1]{Equation~\ref{eq:#1}}

\newcommand{\reflem}[1]{Lemma~\ref{lem:#1}}
\newcommand{\refrem}[1]{Remark~\ref{rem:#1}}

\newcommand{\refex}[1]{Example~\ref{ex:#1}}

\newcommand{\id}{\mathrm{ID}}

\newcommand{\nbrs}{\texttt{GetNbrs}\xspace}
\newcommand{\steal}{\texttt{StealWork}\xspace}

\newcommand{\scan}{\texttt{SCAN}\xspace}
\newcommand{\scannospace}{\texttt{SCAN}}
\newcommand{\extend}{\texttt{PULL-EXTEND}\xspace}
\newcommand{\extendnospace}{\texttt{PULL-EXTEND}}
\newcommand{\join}{\texttt{PUSH-JOIN}\xspace}

\newcommand{\sink}{\texttt{SINK}\xspace}
\newcommand{\sinknospace}{\texttt{SINK}}

\newcommand{\lrbu}{\texttt{LRBU}\xspace}
\newcommand{\lru}{\texttt{LRU}\xspace}
\newcommand{\cinsert}{\texttt{Insert}\xspace}
\newcommand{\creserve}{\texttt{Seal}\xspace}
\newcommand{\cfree}{\texttt{Release}\xspace}
\newcommand{\cget}{\texttt{Get}\xspace}
\newcommand{\ccontains}{\texttt{Contains}\xspace}

\newcommand{\dtgo}{\textsf{GO}\xspace}
\newcommand{\dtlj}{\textsf{LJ}\xspace}
\newcommand{\dtor}{\textsf{OR}\xspace}
\newcommand{\dtuk}{\textsf{UK}\xspace}
\newcommand{\dteu}{\textsf{EU}\xspace}
\newcommand{\dtfs}{\textsf{FS}\xspace}
\newcommand{\dtcw}{\textsf{CW}\xspace}

\newcommand{\unit}{\mathcal{U}\xspace}
\newcommand{\ord}{\mathcal{O}\xspace}
\newcommand{\alg}{\mathcal{A}\xspace}
\newcommand{\comm}{\mathcal{C}\xspace}
\newcommand{\idx}{\kw{Idx}\xspace}
\newcommand{\wopt}{\kwnospace{wco}\xspace}
\newcommand{\leaves}{\mathcal{L}\xspace}

\begin{document}


\title{HUGE: An Efficient and Scalable Subgraph Enumeration System (Complete Version)}


\author{Zhengyi Yang}
\affiliation{\institution{The University of New South Wales}}
\email{zyang@cse.unsw.edu.au}

\author{Longbin Lai}
\affiliation{\institution{Alibaba Group}}
\email{longbin.lailb@alibaba-inc.com}

\author{Xuemin Lin}
\affiliation{\institution{The University of New South Wales}}
\email{lxue@cse.unsw.edu.au}

\author{Kongzhang Hao}
\affiliation{\institution{The University of New South Wales}}
\email{khao@cse.unsw.edu.au}

\author{Wenjie Zhang}
\affiliation{\institution{The University of New South Wales}}
\email{zhangw@cse.unsw.edu.au}

\fancyhead{}

\begin{abstract}
Subgraph enumeration is a fundamental problem in graph analytics, which aims to find all instances of a given query graph on a large data graph. In this paper, we propose a system called \page to efficiently process subgraph enumeration at scale in the distributed context.
\page features 1) an optimiser to compute an advanced execution plan without the constraints of existing works; 2) a hybrid communication layer that supports both pushing and pulling communication; 3) a novel two-stage execution mode with a lock-free and zero-copy cache design, 4) a BFS/DFS-adaptive scheduler to bound memory consumption, and 5) two-layer intra- and inter-machine load balancing. 
\page is generic such that all existing distributed subgraph enumeration algorithms can be plugged in to enjoy automatic speed up and bounded-memory execution.
\end{abstract}

\maketitle

\section{Introduction}
\label{sec:intro}
Subgraph enumeration is a fundamental problem in graph analytics that aims to find all subgraph instances of a data graph that are isomorphic to a query graph. Subgraph enumeration is associated with a lot of real-world applications \cite{protein,biomolecular,chemical,recommend,www,net-analy,motif,GPAR,qa-rdf,wang2019vertex,StructSim}. 
Additionally, it is a key operation when querying graph databases such as Neo4j \cite{neo4j}, and also plays a critical role in graph pattern mining (GPM) systems \cite{Arabesque,AutoMine,Fractal,Peregrine}. 

With the growth in graph size nowadays \cite{ubiquity} and the NP-hardness \cite{np-complete} of subgraph enumeration, researchers have devoted enormous efforts into developing efficient and scalable algorithms in the distributed context \cite{crystaljoin,wco-join,twin-twig,edge-join,star-join,seed,psgl,RADS,BENU}. 


\begin{table}[]
\caption{Results of a square query over the \dtlj graph in a local 10-machine cluster, including total time ($T$), computation time ($T_R$), communication time ($T_C = T - T_R$), total data transferred ($C$), and peak memory usage ($M$) among all machines.}
\label{tab:intro_exp}
\small
\setstretch{0.9}

\begin{tabular}{|c|c|c|c|c|c|c|}
\hline
Comm. Mode & Work & \begin{tabular}[c]{@{}c@{}}$T$(s)\end{tabular} & \begin{tabular}[c]{@{}c@{}}$T_R$(s)\end{tabular} & \begin{tabular}[c]{@{}c@{}}$T_C$(s)\end{tabular} & \begin{tabular}[c]{@{}c@{}}$C$(GB)\end{tabular} & \begin{tabular}[c]{@{}c@{}}$M$(GB)\end{tabular}\\ \hline
\multirow{2}{*}{Pushing} & \seed & $1536.6$ & $343.2$ & $1193.4$ & $537.2$ & $42.3$ \\ \cline{2-7} 
 & \bigjoin & $195.9$ & $122.1$ & $73.8$ & $534.5$ & $14.3$ \\ \hline 
\multirow{2}{*}{Pulling} & \benu & $4091.7$ & $3763.2$ & $328.5$ &  $25.3$ & $1.3$ \\ \cline{2-7} 
 & \rads & $2643.8$ & $2478.7$ & $165.1$  & $452.7$ & $19.2$ \\ \hline
\textbf{Hybrid} & \begin{tabular}[c]{@{}c@{}}\page\end{tabular} & $52.3$ & $51.5$ & $0.8$ & $4.6$ & $2.2$ \\ \hline
\end{tabular}
\vspace*{-2\baselineskip}
\end{table}

\stitle{Motivations.} 
The efficiency and scalability of distributed subgraph enumeration are jointly determined by three perspectives: \emph{computation}, \emph{communication} and \emph{memory management} \cite{patmat-exp, RADS}. 
However, existing works \cite{seed, wco-join, BENU, RADS} \footnote{We mainly discuss four representative works here, while the others are in \refsec{related}.} fail to demonstrate satisfactory performance for all three perspectives. 
To verify, we conduct an initial experiment by running the square query ($\bm{\square}$) over the popular benchmark graph \dtlj \cite{patmat-exp}. 
The results\footnote{The results may differ from the original reports of \benu and \rads because we use better implementations of the join-based algorithms \cite{patmat-exp} (\refsec{exp}).} are shown in \reftable{intro_exp}.

\seed \cite{seed} and \bigjoin \cite{wco-join} are join-based algorithms that adopt the \emph{pushing} communication mode, which communicates by sending data from the host machine to remote destinations. 
In general, \seed processes subgraph enumeration via a series of binary joins, each joining the matches of two sub-queries using the conventional hash join. \bigjoin \cite{wco-join} follows the worst-case optimal (\wopt) join algorithm \cite{Ngo-join}, which extends the (intermediate) results one vertex at a time by intersecting the neighbours of all its connected vertices. Both algorithms are scheduled in a breadth-first-search (BFS) order \cite{crystaljoin} in order to fully utilize the parallel computation, which in turn requires materializing and transferring (via pushing) enormous intermediate results. Such design choices can cause high tension on both communication and memory usage, as shown in \reftable{intro_exp}.

While noticing the enormous cost from pushing communication, \benu \cite{BENU} and \rads \cite{RADS} exploit a \textit{pulling} design. \benu has been developed to pull (and cache) graph data from a distributed key-value store (e.g. Cassandra \cite{cassandra}). On each machine, it embarrassingly parallelises a sequential depth-first-search (DFS)-based program \cite{ullmann} to compute the matches. Such pulling design substantially reduces \benu's communication volume, which, however, does not shorten its communication time accordingly. The main culprit is the large overhead of pulling (and accessing cached) data from the external key-value store. 
Additionally, while the use of DFS strategy results in low memory consumption, it can suffer from low CPU utilisation \cite{dfs}. The above shortages jointly reduce the computing efficiency of \benu. To support a more efficient pulling design, \rads has developed its own compute engine without external dependency. Observe that the matches of a star (a tree of depth) rooted on a vertex can be enumerated from its neighbours \cite{twin-twig}. Instead of transferring the intermediate results, the join that involves a star can be computed locally after pulling to the host machine the remote vertices with their neighbours. However, to facilitate such a pulling design, \rads is coupled with a \starjoin-like \cite{star-join} execution plan that has already been shown to be sub-optimal \cite{seed, wco-join}, which leads to poor performance of \rads in all perspectives.

\stitle{Challenges.} We distil three impact factors that jointly affect the three perspectives of distributed subgraph enumeration, namely \emph{execution plan}, \emph{communication mode}, and \emph{scheduling strategy}.

\sstitle{Execution plan.} Existing works derive their ``optimal'' execution plans, while none can guarantee the best performance by all means, as evidenced by \cite{patmat-exp} and the results in \reftable{intro_exp}. The main reason is that these works achieve optimality in a rather \emph{specific} context subject to the join algorithm and communication mode. For example, \seed is optimal among the hash-join-based algorithms \cite{star-join, twin-twig, seed}, while \bigjoin's optimality follows the \wopt-join algorithm. The optimal plan of \rads is computed specifically for its pulling-based design. We argue that an optimal execution plan should lie in a more generic context without the constraints of existing works, which clearly makes it challenging to define and compute.

\sstitle{Communication mode.} While pulling mode can potentially reduce communication volume, it is non-trivial to make it eventually improve overall performance. Regarding design choice, it is not an option to blindly embrace the pulling design, as \rads has practised, without considering its impact on the execution plan. Regarding implementation, it is infeasible to directly utilise an external infrastructure that can become the bottleneck, as \benu has encountered.

\sstitle{Scheduling strategy.} Although DFS strategy has small memory requirement, it can suffer from low network and CPU utilisation. To saturate CPU power (parallelism), BFS strategy is more widely used for distributed subgraph enumeration. However, it demands a large memory to maintain enormous intermediate results. \emph{Static} heuristics such as batching \cite{wco-join} and region group \cite{RADS} are used to ease memory tension by limiting the number of initially matched (pivot) vertices/edges. Nevertheless, such static heuristics all lack in a tight bound and can perform poorly in practice. In our experiment (\refsec{exp}), we have observed out-of-memory errors from the static heuristics, even while starting with one pivot vertex/edge.  

\stitle{Our Solution and Contributions.} 
In this paper, we take on all aforementioned challenges
by presenting a system called \page, short for pushing/pulling-\textbf{H}ybrid s\textbf{U}b\textbf{G}raph \textbf{E}numeration system. Specifically, we make the following contributions:

\sstitle{(1) Advanced execution plan.} We study to break down an execution plan of subgraph enumeration into the logical and physical aspects. Logically, we express all existing works \cite{twin-twig,edge-join,star-join,seed,wco-join, BENU, RADS} in a uniform join-based framework. As a result, these works can be readily plugged into \page to enjoy automatic performance improvement. Physically, we carefully consider the variances of join algorithms (hash join and \wopt join) and communication modes (pushing and pulling) for better distributed join processing. As a result, we are able to search for an optimal execution plan to minimise both communication and computation cost in a more generic context without the constraints of existing works. 


\sstitle{(2) Pushing/pulling-hybrid compute engine.} As the generic execution plan may require both pushing and pulling communication, we develop a hybrid compute engine that efficiently supports dual communication mode. Communication-wise, the dual-mode communication allows the runtime to use either pushing or pulling communication based on which mode produces less cost (according to the optimal plan). As a result, \page can benefit from substantially reduced communication volume, as can be seen from \reftable{intro_exp}, where \page renders the smallest communication volume of $4.6$GB, and the lowest communication time of $0.8$s. Computation-wise, while noticing that cache is the key to efficient pulling-based computation, we devise a new cache structure called \emph{least-recent-batch used} (\lrbu) cache. Together with a two-stage execution strategy, we achieve \emph{lock-free} and \emph{zero-copy} cache access with small synchronisation cost. Additionally, a two-layer intra- and inter-machine work-stealing mechanism is employed for load balancing. Overall, these techniques contribute to \page's superior performance. As shown in \reftable{intro_exp}, \page outperforms \seed, \bigjoin, \benu and \rads by 29.4$\times$, 3.7$\times$, 78.2$\times$, 50.6$\times$, respectively. 

\sstitle{(3) BFS/DFS-adaptive scheduler.} To manage memory usage without sacrificing computing efficiency, we introduce a BFS/DFS-adaptive scheduler to dynamically control the memory usage of subgraph enumeration. It adopts BFS-style scheduling whenever possible to fully leverage parallelism and adapts dynamically to DFS-style scheduling if the memory usage exceeds a constant threshold. With the scheduler, we prove that \page achieves a tight memory bound of $O(|V_q|^2\cdot D_G)$ for a subgraph enumeration task, where $|V_q|$ is the number of query vertices and $D_G$ is the maximum degree of the data graph. As a result, \page uses only slightly more memory than \benu (purely DFS scheduling) in \reftable{intro_exp}, while achieving the best performance among the competitors.

\sstitle{(4) In-depth experiment.} We conduct extensive experiments on 7 real-world graphs. Results show the effectiveness of our techniques. To highlight, \page outperforms previously best pulling-based algorithm by up to $105\times$, and the best join-based algorithm by up to $14\times$, with considerably much less communication and memory usage.

\stitle{Paper Organization.} The rest of this paper is organized as follows. \refsec{pre} introduces preliminaries. \refsec{framework} presents \page's optimiser. 
We present implementation details of \page in \refsec{impl} and how computation is scheduled in \page to achieve bounded-memory execution in \refsec{schedule}. We discuss the potential applications of \page in \refsec{applications}. Empirical evaluations are in \refsec{exp}, followed by related work in \refsec{related} and conclusion in \refsec{conclusion}.
\section{Preliminaries}
\label{sec:pre}

\stitle{Graph Notations.} We assume both the data graph and query graph are unlabelled, undirected, and connected\footnote{Our techniques can seamlessly support directed and labelled graph.}. A graph is a tuple $g = (V_g, E_g)$, where $V_g$ is the vertex set and $E_g \subseteq V_g \times V_g$ is the edge set of $g$. For a vertex $\mu \in V_g$, we use $\mathcal{N}_g(\mu)$ to denote the neighbours of $\mu$, and $d_g(\mu) = |\mathcal{N}_g(\mu)|$ to denote the degree of $\mu$. The average and maximum degree of $g$ is denoted as $\overline{d_g}$ and $D_g$, respectively. 
Each vertex $v\in V_g$ is assigned with an unique integer ID from $0$ to $|V_g|-1$ denoted as $\id(v)$. A \emph{star}, denoted as $(v; \leaves)$, is a tree of depth 1 with $v$ as the root and $\leaves$ as the leaf vertices. 
A \emph{subgraph} $g'$ of $g$, denoted $g' \subseteq g$, is a graph such that $V_{g'}\subseteq V_g$ and $E_{g'}\subseteq E_g$. A subgraph $g'$ is an \emph{induced subgraph} of $g$ if and only if $\forall \mu,\mu' \in V_{g'}, e=(\mu,\mu')\in E_g$ it holds that $e\in E_{g'}$. We denote $g = g_1 \cup g_2$ for merging two graphs, where $V_g = V_{g_1} \cup V_{g_2}$ and $E_g = E_{g_1} \cup E_{g_2}$.

\stitle{Subgraph Enumeration.} Two graphs $q$ and $g$ are \emph{isomorphic} if and only if there exists a bijective mapping $f:V_q \rightarrow V_g$ such that $\forall (v,v') \in E_q, (f(v),f(v')) \in E_g$. Given a query graph $q$ and a data graph $G$, the task of \emph{subgraph enumeration} is to enumerate all subgraphs $g$ of $G$ such that $g$ is isomorphic to $q$. Each isomorphic mapping from $q$ to $g$ is called a \emph{match}. 
By representing the query vertices as $\{v_1,v_2, \dots ,v_n\}$, we can simply denote a match $f$ as $\{u_{k_1}, u_{k_2}, \ldots, u_{k_n}\}$, where $f(v_i) = u_{k_i}$ for $1 \leq i \leq n$. We call a subgraph $q'$ of $q$ a \emph{partial query}, and a match of $q'$ a \emph{partial match}. Given a query graph $q$ and data graph $G$, we denote the result set of subgraph enumeration as $\mathbb{R}_G(q)$, or $\mathbb{R}(q)$ if it is clear. 

As a common practice, we apply the method of \emph{symmetry breaking} \cite{symmetry-breaking} to avoid duplicated enumeration caused by automorphism (an isomorphism from a graph to itself).

\comment{
\stitle{Symmetry Breaking.} . Specifically, we assign a total order, denoted as $\prec$, among all vertices in the data graph $G$ and a partial order, denoted as $<$, among some pairs of vertices in the query graph $q$. An \emph{order-preservation constraint} is therefore enforced in the match to break symmetry. We relabel the IDs of the data vertices, such that the smaller-degree vertex is assigned with smaller ID (tie breaks by the original IDs). Given two data vertices $v,v' \in V_G$, we have $v \prec v'$ if and only if $\id(v)<\id(v')$.
}

\stitle{Graph Storage.} We \emph{randomly} partition a data graph $G$ in a distributed context as most existing works \cite{patmat-exp,seed,twin-twig,wco-join,crystaljoin}. For each vertex $\mu \in V_G$, we store it with its adjacency list $(\mu;\mathcal{N} (\mu))$ in one of the partitions. We call a vertex that resides in the local partition as a \emph{local vertex}, and a \emph{remote vertex} otherwise. 

\stitle{Ordered Set.} An \emph{ordered set} is a pair $\hat{S}=(S,Ord)$, where $S$ is a set and $Ord$ is the corresponding map of ordering, which maps each element in S to an integer. For $s_1,s_2\in \hat{S}$, we say $s_1\leq s_2$ if and only if $Ord(s_1)\leq Ord(s_2)$. Besides, we use $\min(\hat{S})$ and $\max(\hat{S})$ to denote an element in $\hat{S}$ with the smallest and largest order, respectively.

\stitle{Remote Procedure Call.} A \emph{remote procedure call} (RPC) \cite{rpc} is when a computer program calls a procedure to execute in a different address space. We refer to the caller as \emph{client} and the executor as \emph{server}. 
The form of request–response interaction allows RPC to be naturally adopted for pulling communication.

\section{Advanced Execution Plan}
\label{sec:framework}
In this section, we first show that existing works can fit uniformly into a \emph{logical} join-based framework. Then we discuss two primary physical settings for distributed join processing. We eventually propose a dynamic-programming-based optimiser to compute the optimal execution plan for subgraph enumeration.   


\subsection{A Logical Join-based Framework}
\label{sec:plan}
It is known that subgraph enumeration can be expressed as a multi-way join of some basic structures called join units (e.g. edges, stars) \cite{seed}. Given a query graph $q$ and a data graph $G$, and a sequence of join units $\{q_1, q_2, \ldots q_k\}$, such that $q = q_1 \cup q_2 \cup \cdots q_k$, we have
\begin{equation}
    \small
    \label{eq:common_join}
    \mathbb{R}_G(q) = \mathbb{R}_G(q_1) \Join \mathbb{R}_G(q_2) \Join \cdots \Join \mathbb{R}_G(q_k).
\end{equation}


Logically speaking, existing works all solve the above join via multiple rounds of \emph{two-way} joins, with the variances in join unit ($\unit$) and join order ($\ord$). For simplicity, we represent a two-way join $\mathbb{R}(q') = \mathbb{R}(q'_l) \Join \mathbb{R}(q'_r)$ as a 3-tuple $(q', q'_l, q'_r)$. The join order is an ordered sequence of two-way joins $(q', q'_l, q'_r)$ (where $q', q'_l, q'_r \subseteq q$), with its last element being $(q, q_l, q_r)$. 

\starjoin \cite{star-join} pioneers the idea of using stars as the join unit, 
as well as the left-deep join order $\ord_{ld}$, in which it requires that $q'_r$ is a join unit for each $(q',q'_l,q'_r) \in \ord_{ld}$. \seed \cite{seed} further allows using clique (a complete graph), in addition to stars, as the join unit, after maintaining extra index (triangle index). Moreover, \seed replaces the prior practice of left-deep join order with bushy join, which removes the constraint that each $q'_r$ is a join unit, and hence covers a more complete searching space for an optimal execution plan. 

\stitle{BiGJoin} We uncover the connections between \bigjoin \cite{wco-join} and the join-based framework as follows. \bigjoin is based on the \wopt join algorithm \cite{Ngo-join}. It matches the query graph one vertex at a time in a predefined order. Let the matching order be $V_q=\{v_1,v_2,\dots,v_n\}$. The execution plan starts from an empty set, and computes the matches of $\{v_1,\dots,v_i\}$ in the $i^{\mathrm{th}}$ round. Let a partial match after the $i^{\mathrm{th}}$ round be $p=\{u_{k_1}, u_{k_2}, \ldots, u_{k_i}\}$ for $i < n$, \bigjoin expands the results in the $(i+1)^{\mathrm{th}}$ round by matching $v_{i+1}$ with $u_{k_{i+1}}$ for $p$ if and only if $\forall_{1\leq j\leq i}(v_j,v_{i+1})\in E_q, (u_{k_j},u_{k_{i+1}})\in E_G$. The candidate set of $v_{i+1}$, denoted as $\mathbb{C}(v_{i+1})$ can be computed by the following intersection
\begin{equation}
    \small
    \label{eq:wco}
    \mathbb{C}(v_{i + 1}) = \cap_{\forall_{1 \leq j \leq i} \land (v_j, v_{i + 1}) \in E_q} \mathcal{N}_G(u_{k_j}). 
\end{equation}


\begin{definition}
\label{def:complete_star_join}
A two-way join $(q', q'_l, q'_r)$ is a \emph{complete star join} if and only if $q'_r$ is a star $(v'_r; \leaves)$ (w.l.o.g. \footnote{Note that as join operation is commutative, the condition also applies to $q'_l$. Thereafter, we always present $q'_r$ without loss of generality (w.l.o.g.) in this paper.}) and $\leaves \subseteq V_{q'_l}$.
\end{definition}

We show how \bigjoin can be expressed in the join-based framework. Let $q_{i}' = q_1 \cup \cdots \cup q_{i}$. The procedure of \bigjoin is equivalent to the joins following the left-deep order $\ord_{ld}$, where the $i^\mathrm{th}$ element of $\ord_{ld}$ is $(q'_{i+1}, q'_i, q_{i+1})$, and it further satisfies that $q_{i}'$ is an \emph{induced} subgraph of $q$, and $(q'_{i+1}, q'_i, q_{i+1})$ is a complete star join with $q_{i+1} = (v_{i+1}; \leaves_{i+1})$.

\begin{figure*}[ht]
    \centering
    \begin{subfigure}[b]{.08\textwidth}
            \centering
            \includegraphics[width=.9\textwidth]{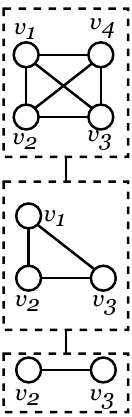}
            \caption{}
            \label{fig:4-clique-bigjoin}
    \end{subfigure} %
    ~
    \begin{subfigure}[b]{.24\textwidth}
            \centering
            \includegraphics[width=.9\textwidth]{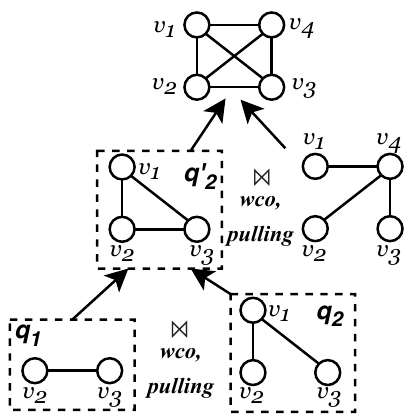}
            \caption{}
            \label{fig:4-clique-execution-plan}
    \end{subfigure} %
    ~
    \begin{subfigure}[b]{.16\textwidth}
            \centering
            \includegraphics[width=.9\textwidth]{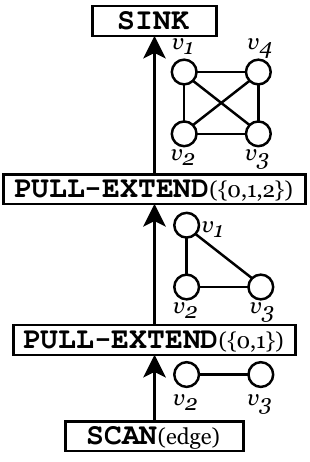}
            \caption{}
            \label{fig:4-clique-dataflow}
    \end{subfigure} %
    ~
    \begin{subfigure}[b]{.2\textwidth}
            \centering
            \includegraphics[width=.9\textwidth]{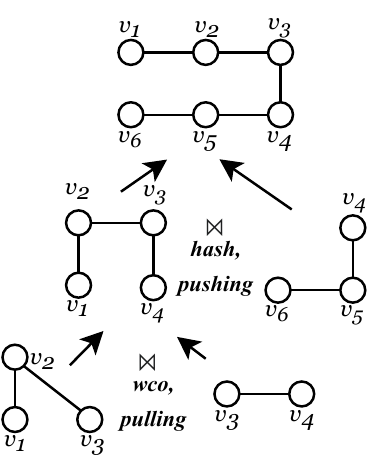}
            \caption{}
            \label{fig:5-path-execution-plan}
    \end{subfigure} %
    ~
    \begin{subfigure}[b]{.32\textwidth}
            \centering
            \includegraphics[width=.9\textwidth]{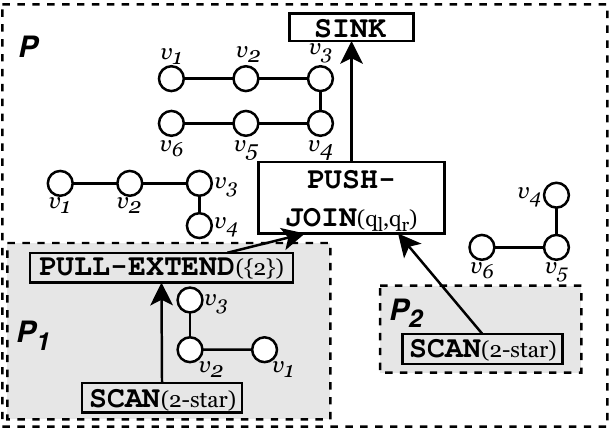}
            \caption{}
            \label{fig:5-path-dataflow}
    \end{subfigure} %
    \caption{Example execution plans and dataflow graphs, where (a) the \bigjoin plan of 4-clique; (b) the execution plan by \refalg{dp-opt} of (a); (c) the dataflow graph of (b); (d) an execution plan by \refalg{dp-opt} of 5-path; (e) the dataflow graph of (d). }
    \label{fig:example_plans}
\end{figure*}

\begin{example}
Given a \bigjoin's execution of a 4-clique in \reffig{4-clique-bigjoin}, we present its corresponding execution plan in \page in\reffig{4-clique-execution-plan}, where each vertex extension is expressed as a complete star join. As an example, the extension of $v_1$ from an edge $(v_2, v_3)$ is equivalent to the complete star join of $(q'_2, q_1, q_2)$. 
\end{example}

\stitle{BENU.} \benu stores the whole graph data in a distributed key-value store (e.g. Cassandra \cite{cassandra}). After pulling and caching required data locally, \benu runs a depth-first-search (DFS) -based subgraph isomorphism program (e.g. \cite{ullmann} in each machine). The program matches the query vertices along with the DFS tree, and checks the connections between the currently visited vertex and the already matched vertices. This is equivalent to \bigjoin's \wopt-join procedure with the DFS order as matching order and thus fits into the join-based framework.

\stitle{RADS.} \rads employs a multi-round ``star-expand-and-verify paradigm''. It first expands the partial results by a star rooted at one of the already matched vertices in each round. This is logically equivalent to joining the partial results with a star $(v; V_f)$ where $v$ must belong to the join key. Verification is then applied to filter out infeasible results based on edge verification index. This step is logically equivalent to joining the partial results with data edges (i.e. 1-star).  As a result, \rads fits into the join-based framework.

\subsection{Physical Join Processing}
\label{sec:physical_join}

Given the join-based framework, the performance of subgraph enumeration is further determined by how the join is physically processed. Here, we consider two physical settings for distributed join processing, namely, join algorithm ($\alg$) and communication mode ($\comm$). Let an arbitrary join be $(q', q'_l, q'_r)$.

\stitle{Join Algorithm.} While distributed join is well studied \cite{join-exp, Ngo-join, multi-join}, we focus on the algorithms adopted for subgraph enumeration. A distributed hash join algorithm is the foundation of \cite{star-join, psgl, twin-twig, seed, RADS}. Given $\mathbb{R}(q'_l)$ and $\mathbb{R}(q'_r)$, hash join typically shuffles $\mathbb{R}(q'_l)$ and $\mathbb{R}(q'_r)$ based on the join key of $V_{q'_l} \cap V_{q'_r}$. Thus, hash join needs to fully materialize both $\mathbb{R}(q'_l)$ and $\mathbb{R}(q'_r)$, which can be wasteful as only a part of $\mathbb{R}(q'_l)$ and $\mathbb{R}(q'_r)$ can produce feasible results. In the case that the above join is a \emph{complete star join}, the more efficient \wopt join algorithm can be used by processing the intersection in \refeq{wco}. Instead of blindly materializing the data for join, \wopt join can benefit from the worst-case optimal bound \cite{Ngo-join} to only materialize necessary data. 

\stitle{Communication Mode.} 
It is straightforward to process the distributed join in the pushing mode. For hash join, we shuffle $\mathbb{R}(q'_l)$ and $\mathbb{R}(q'_r)$ by \emph{pushing} the results to the remote machines indexed by the join key. 
For \wopt join with $p_r = (v'_r; \leaves)$, we push each $f \in \mathbb{R}(q'_l)$ to the remote machine that owns $f(v)$ continuously for each $v \in \leaves$ to process the intersection. In certain join scenario, we may leverage the pulling communication mode to process the join, in which a host machine rather \emph{pull}s the graph data than \emph{push}es the intermediate results. We have the following observation:

\begin{property}
\label{ppt:pulling}
The pulling communication can be adopted if $q'_r$ is a star $(v'_r; \leaves)$, and the join satisfies one of the following conditions: (C1) $v'_r \in V_{q_l}$; and (C2) the join is a complete star join.
\end{property}

Let $f$ be a match of $q'_l$, and $u'_r = f(v'_r)$. Regarding C1, after pulling $\mathcal{N}_G(u'_r)$ from the machine that owns $u'_r$, the join can be locally processed with the matches of $q'_r$ (rooted on $u'_r$) enumerated as $|\leaves|$-combinations over $\mathcal{N}_G(u'_r)$ \cite{twin-twig}; regarding C2, while \wopt join must be used, after pulling $\mathcal{N}_G(f(v))$ for all $v \in \leaves$ from a remote machine, the intersection (\refeq{wco}) can be locally computed. 

\begin{remark}
\label{rem:pushing_pulling}
Given a join$(q', q'_l, q'_r)$, in the pushing mode, we need to transfer data of size $|\mathbb{R}(q'_l)| + |\mathbb{R}(q'_r)|$ in the case of hash join, and $\overline{d_G} |\mathbb{R}(q'_l)|$ in the case of \wopt join. In the pulling mode, however, we pull \emph{at most} the whole graph data for each machine (i.e. $|E_G|$). Since the size of intermediate results (e.g. $\mathbb{R}(q'_l)$ and $\mathbb{R}(q'_r)$) is usually order-of-magnitude larger than the data graph itself in subgraph enumeration \cite{BENU,AGM,patmat-exp}, pulling can potentially benefit from reduced communication.
\end{remark}

\begin{table}[]
    \small
    \setstretch{0.95}

    \centering
    \caption{Existing works and their execution plans.}
    \begin{tabular}{|l||l|l||l|l|}
    \hline
     & \multicolumn{2}{c||}{Logical} & \multicolumn{2}{c|}{Physical} \\ \hline
      Existing Work & $\unit$ & $\ord$ & $\alg$ & $\comm$  \\
      \hline\hline
        \starjoin \cite{star-join} & star & left-deep & hash join & pushing \\  \hline 
         \seed \cite{seed} & star \& clique & bushy & hash join & pushing \\ \hline 
         \bigjoin \cite{wco-join} & star (limited\footnotemark) & left-deep & \wopt join & pushing \\ \hline 
         \benu \cite{BENU} & star (limited) & left-deep & \wopt join & pulling \\ \hline
         \rads \cite{RADS} & star & left-deep & hash join & pulling \\ \hline 
    \end{tabular}
    \label{tab:existing_works}
\end{table}

\footnotetext{It only accepts limited form of stars as discussed before.}

\subsection{Optimal Execution Plan}
\label{sec:optimal_plan}
We summarize existing works and their execution plans in \reftable{existing_works}, it is clear that they are subject to specific settings of join algorithm and communication mode. 
To pursue an optimal execution plan in a more generic context, we break down an execution plan $P$ for subgraph enumeration into the logical settings of join unit ($\unit$) and join order ($\ord$), as well as physical settings of join algorithm ($\alg$) and communication mode ($\comm$). Specially, we call $L = (\unit, \ord)$ as the logical plan. We detail the settings of \page as follows. By default, we use stars as the join unit, as our system does not assume any index data. 
We use the bushy join order \cite{bushy-join} as it covers more complete searching space. Given an arbitrary join of $(q', q'_l, q'_r)$, we configure its physical settings according to \refsec{physical_join} as:
\begin{equation}
    \small
    \label{eq:settings}
    (\alg, \comm)=\left\{
    \begin{array}{ll}
        (\text{\wopt join, pulling}), & \text{if it is a complete star join},  \\
        (\text{hash join, pulling}), & \text{if } q'_r \text{ is a star } (v'_r; \leaves) \land v'_r \in V_{q_l},  \\
        (\text{hash join, pushing}), & \text{otherwise.}
    \end{array} \right.
\end{equation}

We are now ready to present \refalg{dp-opt} to compute the optimal execution plan for \page with the aim of minimizing both computation and communication cost. 

\begin{figure}[t!]
\removelatexerror
\begin{algorithm*}[H]
    \small
    \setstretch{0.9}
    
    \SetKwFunction{cost}{Cost}
    \SetKwFunction{settings}{ConfigureJoin}
    \SetKwFunction{recover}{RecoverJoinOrder}
        $M_{plan}\leftarrow \{\}, M_{cost}\leftarrow \{\}$ \;
        \For{$n\leftarrow 3 \dots |V_q|$}{
            \ForAll{connected subgraph $q' \subseteq q\ s.t.\ |V_{q'}| = n$}{
                \textbf{if} $q'$ is a join unit \textbf{then} $M_{cost}[q'] \leftarrow |\mathbb{R}(q')|$ \;
                \textbf{else} \ForAll{connected subgraphs $q'_l,q'_r \subset q'\ s.t.\ q'_l\cup q'_r=q' \land E_{q'_l} \cap E_{q'_r} = \emptyset$}{
                        $C \leftarrow M_{cost}[q_l'] + M_{cost}[q_r'] + |\mathbb{R}(q')|$ \;
                        \If {$(q', q'_l,  q'_r)$ applies pulling by \refeq{settings} }{
                            $C \leftarrow C + k |E_G|$  \;
                        }
                        \textbf{else}  $C \leftarrow C + |\mathbb{R}(q'_l)| +  |\mathbb{R}(q'_r)|$ \;
                        \If {$M_{cost}[q'] = \emptyset$ or $M_{cost}[q'] > C$} {
                            $M_{cost}[q'] \leftarrow C$ ;
                            $M_{plan}[q'] \leftarrow (q'_l, q'_r)$ \;
                        }
                    }
            }
        }
        $O \leftarrow \recover(M_{plan})$ \;
        $S \leftarrow \settings(O)$ \;
    
    \Return{$(\ord, S)$}
    
  \caption{\kw{OptimalExecutionPlan}($q$).}
  \label{alg:dp-opt}
\end{algorithm*}
\end{figure}

The optimiser starts by initializing two empty maps $M_{plan}$ and $M_{cost}$. Given a sub-query $q'$, $M_{plan}$ stores the mapping from $q'$ to be best-so-far join that produces $q'$, and $M_{cost}$ records the corresponding cost (line~1). Note that we are only interested in the non-trivial case where $q$ is not a join unit. The program goes through the searching space from smaller sub-queries to larger ones (line~2). For a sub-query $q'$, no join is needed if it is a join unit, and we record its computation cost as $|\mathbb{R}(q')|$ (line~4) that can be estimated using the method such as \cite{seed,graphflow,g-care}. Otherwise, the optimiser enumerates all pairs of sub-queries $q'_l$ and $q'_r$ that can be joined to produce $q'$ (line~5). The cost of processing the join is computed in line~7-9, which consists of the cost of processing $q'_l$ and $q'_r$, the computation cost of $q'$ that is $|\mathbb{R}(q')|$, and the communication cost of the join. If pulling mode is configured, the communication cost is at most $k |E_G|$ (line~8), where $k$ is the number of machines in the cluster (\refrem{pushing_pulling}); otherwise, the cost is equal to the shuffling cost of $q'_l$ and $q'_r$, that is $|\mathbb{R}(q'_l)| + |\mathbb{R}(q'_r)|$ (line~9). If $q'$ has not been recorded in $M_{cost}$, or the recorded cost is larger than the current cost $C$, the new cost and join will be updated to the corresponding entries (line~11). Finally, the optimiser recovers the join order $\ord$ from $M_{plan}$ and configures the physical settings according to \refeq{settings} for each join in $\ord$ (line~12). 


\begin{example} 
\label{ex:optimal-plan}
\reffig{4-clique-execution-plan} illustrates the optimal execution plan for the 4-clique. In \reffig{5-path-execution-plan}, we further show the optimal execution plan of a 5-path. The two joins are processed via pulling-based \wopt join and pushing-based hash join, respectively, which demonstrates the need of both pushing and pulling communication. Note that such a plan reflects the works \cite{empty-headed, graphflow} that mix hash join and \wopt join in a hybrid plan space \cite{wco-join}. Nevertheless, these works are developed in a sequential context where computation is the only concern, while we target the distributed runtime that further considers the best communication mode. In the experiment, we show that our optimal execution plan renders better performance than \cite{empty-headed,graphflow}.
\end{example}

 
 
\begin{remark}
\label{rem:plug-in-exist}
With the separation of logical and physical settings, we allow users to directly feed existing logical plans into the optimiser, and the optimiser will only configure the physical settings for each join. Even with the same logical plan, we shall see from the experiment (\refsec{exp}) that \page achieves much better performance due to the optimal settings of join algorithm and communication mode, together with the other system optimisations to be introduced. In this sense, \emph{existing works can be plugged into \page via their logical plans to enjoy immediate speedup and bounded memory consumption}.
\end{remark}

\section{The \page Compute Engine}
\label{sec:impl}

Most existing works have been developed on external big-data engines such as Hadoop \cite{hadoop} and Timely dataflow engine \cite{naiad}, or distributed key-value store such as HBase \cite{hbase} and Cassandra \cite{cassandra}. 
Big-data engines typically do not support pulling communication. Distributed key-value store, however, lacks support of pushing communication, and can become the bottleneck due to large overhead. Thus, they cannot be adopted to run the execution plan in \refsec{framework} that may require both pushing and pulling communication. 

We implement our own pushing/pulling-hybrid compute engine for \page. In this section, we introduce the architecture of the engine, the dataflow computation model, and the distributed join processing, especially the novel pulling-based extend operator. 


\subsection{Architecture}
\label{sec:arch}
\page adopts a shared-nothing architecture in a $k$-machine cluster. There launches a \page runtime in each machine as shown in \reffig{arch}. We briefly discuss the following components, while leaving \emph{Cache} and \emph{Scheduler} to \refsec{pull} and \refsec{schedule}, respectively.

\stitle{RPC Server:} RPC server is used to answer incoming requests from other machines. The server supports two RPCs - \nbrs and \steal. \nbrs takes a list of vertices as its arguments and returns their neighbours. Note that the requested vertices must reside in the current partition. \steal can steal unprocessed tasks locally and send them to a remote machine for load balancing.
	
\stitle{RPC Client:} An RPC client establishes connections with other machines to handle RPC communication. RPC requests will be sent through RPC client whenever RPCs 
are called locally.

\stitle{Router:} The router pushes data to other machines. It manages TCP streams connected to remote machines, with a queue for each connection. The data placed in the queue will be transferred to the corresponding machine based on its routing index (e.g. join keys).

	
\stitle{Worker:} Upon starting, the \page runtime initializes a worker pool containing certain number of workers. While an operator is scheduled to run, it will be assigned to the worker pool and executed by all workers to perform the de-facto computation. Each worker has access to the local partition of the graph, RPC client and the in-memory cache. If requesting a local vertex, it will return its neighbours from the local partition. Otherwise, it either returns the cached value if any, or sends an RPC request through the RPC client to obtain the neighbours, caches them, and returns the neighbours. 


\begin{figure}[t]
    \centering
    \includegraphics[width=0.42\textwidth]{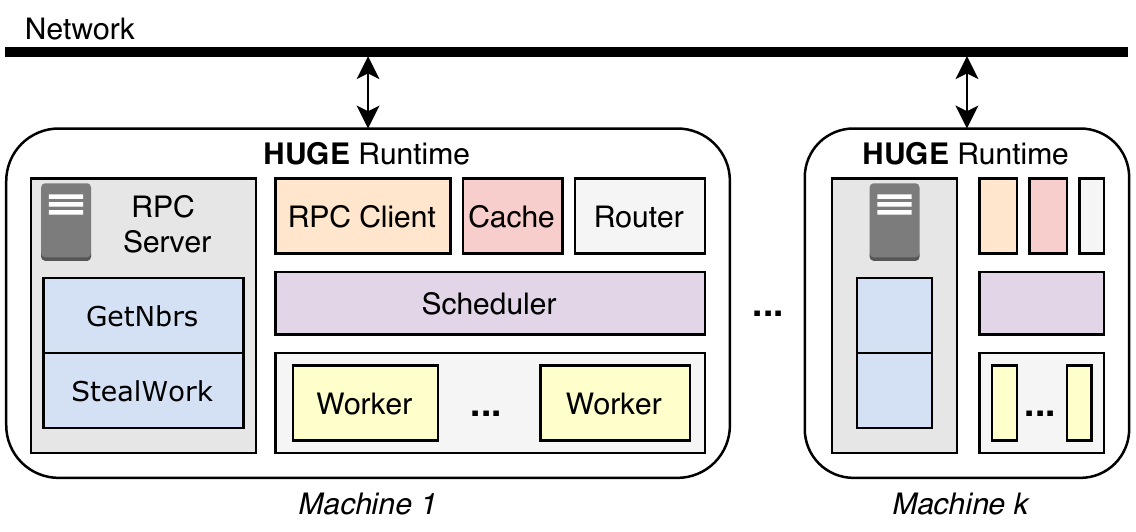}
    \caption{\page Architecture}
    \label{fig:arch}
\end{figure}


\subsection{Dataflow Model}
\label{sec:opt}
We adopt the popular dataflow model \cite{naiad, tensorflow} for \page, where computation is abstracted as a dataflow graph. A dataflow graph is a directed acyclic graph (DAG), in which each vertex is an operator, and the directed edges represent data flows. 
An \textit{operator} is the basic computing unit to run in \page, consisted of a predefined computing instruction, a certain number of inbound channels and one outbound channel. The computation of an operator is driven by receiving data from the inbound channels. Given two operators $O_1$ and $O_2$ that are connected in the dataflow graph, the data produced from the outbound channel of $O_1$ will be fed to one of the inbound channels of $O_2$. We call $O_1$ the \emph{precursor} of $O_2$, and $O_2$ the \emph{successor} of $O_1$. Upon receiving a dataflow, \page distributes it to each machine to drive the computation. 

We introduce four primitive operators necessary to understand this paper, namely \scan, \sink, \join and \extend. More operators can be added to \page to support complex analytical tasks \cite{neo4j, G-Miner} (discussed in \refsec{applications}). A valid dataflow must start from a \scan operator and end with a \sink operator. $\scan(q')$ accepts a join unit $q'$ as its parameter, takes a local partition of the data graph $G$, iterates over the partition, and outputs the matches of $q'$ in the partition. \sink is used to consume the results of subgraph enumeration, via either counting or writing to external I/O devices (e.g. disk). We introduce the semantics of \join and \extend here, and leave the detailed implementation to \refsec{join} and \refsec{pull}. 

\sstitle{PUSH-JOIN.} $\join(q'_l, q'_r)$ processes the pushing-based hash join (\refeq{settings}) of $(q', q'_l, q'_r)$. It configures two inbound channels for the partial results of $\mathbb{R}(q'_l)$ and $\mathbb{R}(q'_r)$ from the precursors. It shuffles (via pushing) $\mathbb{R}(q'_l)$ and $\mathbb{R}(q'_r)$ based on the \emph{join key} of $V_{q'_l} \cap V_{q'_r}$, and then compute the results using local join algorithm. 

\sstitle{PULL-EXTEND.} $\extend(Ext)$ accepts a parameter of \emph{extend index} $Ext=\{d_1,d_2,\dots,d_j\}$. For each input data that is a partial result $f = \{u_1, u_2, \ldots, u_i\}$, the operator extends $f$ by one more vertex as $f' = \{u_1, u_2, \ldots, u_i, u_{i+1}\}$, where the set of possible $u_{i+1}$ is computed as $\bigcap_{k = 1}^{j} \mathcal{N}_G(f[d_k])$. Each $\mathcal{N}_G(f[d_k])$, if not present in local machine, will be \emph{pulled} from the machine that owns $f[d_k]$. 

\begin{figure}[t!]
\removelatexerror
\begin{algorithm*}[H]
    \small
    \setstretch{0.9}
    
    \SetKwFunction{getExt}{GetExt}
    
    $M_{OP}\leftarrow \{\}$; $D\leftarrow \{\}$ \;
    \ForEach{$(q',q'_l,q'_r)$ in $O$}{
        $(\mathcal{A},\mathcal{C})\leftarrow$ the physical setting of $(q',q'_l,q'_r)$ \;
        \If{$\mathcal{C}$ is pushing}{ \tcp{pushing-based hash join}
            $OP_{q'}\leftarrow$\join$(q'_l,q'_r)$ \;
            \ForEach{$q''$ in $q'_l,q'_r$}{
                \If{$q''$ is a join unit}{
                    $D+=\{$\scan$(q'')\rightarrow OP_{q'}\}$ \;
                }
                \Else{
                    $D+=\{M_{OP}[q'']\rightarrow OP_{q'}\}$ \;
                }
            }
        }
        \Else{ \tcp{$\mathcal{C}$ is pulling}
            \If{$\mathcal{A}$ is wco join}{ \tcp{pulling-based wopt join}
                $ext\leftarrow$\getExt{$q'_l,q'_r$} \;
                $OP_q'\leftarrow$\extend$(ext)$ \;
                \If{$q'_l$ is a join unit}{
                    $D+=\{$\scan$(q'_l)\rightarrow OP_{q'}\}$ \;
                }
                \Else{
                    $D+=\{M_{OP}[q'_l]\rightarrow OP_{q'}\}$ \;
                }
            }
            \Else{ 
                \tcp{pulling-based hash join see \refsec{batch}}
                ...
            }
        }
        
        $M_{OP}[q']\leftarrow OP_{q'}$ \;
        
    }
    
    $D+=\{M_{OP}[q]\rightarrow$ \sink$\}$\;
    
    \Return{$(D)$}
    
  \caption{\kw{ExecutionPlanTranslation}($P$).}
  \label{alg:plan-to-df}
\end{algorithm*}
\vspace{-1.6\baselineskip}
\end{figure}

\stitle{Execution Plan Translation.} The \page engine will automatically translate an execution plan given by \refalg{dp-opt} into a dataflow graph. The algorithm is shown in \refalg{plan-to-df}. 
Firstly, in line~1, we initialise an empty map $M_{OP}$ to store the mapping of partial queries to its operator, and an empty dataflow graph $D$. 
\scan operators are installed for each join unit in the execution plan (line~8 and 16), and a \sink operator is added to consume the final results (line~22). Moreover, a pulling-based wopt join and pushing-based hash join (\refeq{settings}) are translated into a \extend and \join operator, respectively. 
For pulling-based hash join, we will show in \refsec{batch} how it will be translated via a series of \extend operators for bounded-memory execution.

\begin{example}
The execution plan in \reffig{4-clique-execution-plan} is translated into the dataflow presented in \reffig{4-clique-dataflow}, in which each pulling-based \wopt join is directly translated to a \extend operator. Similarly, the dataflow of \reffig{5-path-execution-plan} is given in \reffig{5-path-dataflow}, in which the top pushing-based hash join is translated into a \join operator. The \scan and \sink operators are added accordingly for computing the join units (stars) and consuming the final results.
\end{example}


\stitle{Overview of Distributed Execution.}  
In the distributed context, each operator's input data is partitioned to each machine and get processed in parallel. The \scan operator directly reads from the data graph that follows the graph partitioning strategy (\refsec{pre}). The \join operator takes two inputs, which will be hash-partitioned according to the join key. As for \extend and \sink operators, their input data are also the output data of their precursors and are hence partitioned. 

As a common practice of big data engines \cite{spark,storm,naiad}, each operator in \page will process a certain number of data as a \emph{batch} at a time. Thus, a batch of data serves as the minimum data processing unit. Without causing ambiguity, when we present ``an operator processes a batch of data'', we mean that each worker in a machine handles one share 
of the batch in parallel. A barrier is used to guarantee that all workers in a machine are running the same operator to process the same batch of data at any time. Due to load skew, different machines may run different operators unless explicit global synchronisation is enforced. We resolve such load skew via work stealing (\refsec{load_balance}). Depending on the scheduling strategy, the operator will consume certain (at least one) batches of input data in each run. If needed by a remote machine, the output data from an operator will be organised in batches and delegated to the router; otherwise, the data will be placed in the designated buffer to be further processed by the successor as the input.

\subsection{PUSH-JOIN Operator}
\label{sec:join}
The \join operator in \page performs distributed hash-join that shuffles the intermediate results according to the join key. Similar to \cite{hadoop,patmat-exp,MapReduce}, we implement a \emph{buffered distributed hash join}.
It shuffles the intermediate results (via \page's router) with the common join key to the same machine, buffers the received data either in memory or on the disk, and then locally compute the join.

The buffer stage can prevent the memory from being overflowed by either branch of data. We configure a constant buffer threshold, and once the in-memory buffer is full for either branch of the join, we conduct an \emph{external merge sort} on the buffered data via the join keys, and then spill them onto the disk. For join processing, assume that the data is buffered on disk (otherwise is trivial), we can read back the data of each join key in a streaming manner (as the data is sorted,), process the join by conventional nested-loop and write out to the outbound channel. This way, the memory consumption is bounded to the buffer size, which is constant.

\subsection{PULL-EXTEND Operator}
\label{sec:pull}

As mentioned, we implement the \extend operator by pulling communication mode. It requires caching remote vertices for future reuse to reduce the pulling requests via network. 
\benu directly uses a traditional cache structure (e.g. LRU or LFU \cite{TinyLFU}) shared by all workers. We have identified two vital issues that considerably slow down cache access from such a straightforward approach.

\begin{itemize} [leftmargin=*]
    \item Memory copies: Getting a vertex from cache involves at least locating the vertex in the cache, updating the cache position, and finally copying all data (mostly neighbours) of this vertex out. Note that such memory copy is inevitable to avoid dangling pointers in the traditional cache structures, as the memory address of each entry can be changed due to potential replacement. 
    \item Lock: Since the shared cache will be concurrently written and read by multiple workers inside a machine, lock must be imposed on the cache to avoid inconsistency caused by data racing.
\end{itemize}



To address the above issues, we target a \emph{lock-free} and \emph{zero-copy} cache design for \page.  
While there exist works that focus on reducing the lock contention of concurrent cache such as \cite{cache_io}, they are not completely free from locks.
For example, benchmarks \cite{benchmark} show that such design can only achieve about $30\%$ reading performance compared to completely lock-free reads.
Moreover, existing zero-copy techniques \cite{zero-copy-RDMA,zero-copy-EMP,Kafka,zero-copy-web-cache} in distributed computation mainly work to dispatch \textit{immutable buffer} directly to network I/O devices, which cannot be applied to our scenario where the cache structure will be frequently mutated. 
Hence, it requires an innovative design, coupling specifically with the execution of the \extend operator for lock-free and zero-copy cache access. 


\begin{figure}[t!]
\removelatexerror
\begin{algorithm*}[H]
    \small
    \setstretch{0.9}
    
    \SetKwProg{mymethod}{Ref Method}{}{}
    \SetKwProg{mutmymethod}{Mut Method}{}{}
    
    \SetKwFunction{get}{Get}
    \SetKwFunction{contains}{Contains}
    \SetKwFunction{reserve}{Seal}
    \SetKwFunction{myinsert}{Insert}
    \SetKwFunction{free}{Release}
    
    \SetKwFunction{isfull}{CacheIsFull}
    \SetKwFunction{isempty}{IsEmpty}
    \SetKwFunction{pop}{PopSmallest}
    \SetKwFunction{add}{Add}
    \SetKwFunction{remove}{Remove}
    \SetKwFunction{spop}{Pop}
    
    \SetKw{mylist}{neighbours}
    \SetKw{mybool}{Bool}
    
    \KwData{A key-value map $M_{cache}$, an ordered set $\hat{S}_{free}$, a set $S_{sealed}$ }
    
    \mymethod{\get{$vid$}  $\rightarrow $ \mylist }{
        \Return{$M_{cache}[vid]$}
    }
    
    \mymethod{\contains{$vid$} $\rightarrow$ \mybool}{
        \Return{$vid\in M_{cache}$}
    }
    
    \mutmymethod{\myinsert{$vid$,\ $neighbours$}}{
        \If{$\isfull{} \land \lnot \hat{S}_{free}$.\isempty{}}{
            $u\leftarrow \hat{S}_{free}.\pop{}$ ;
            $M_{cache}.\remove{u}$ \;
        }
        $M_{cache}[vid]\leftarrow neighbours$ \;
    }
    
    \mutmymethod{\reserve{$vid$}}{
        $\hat{S}_{free}.\remove{vid}$ ;
        $S_{sealed}.\add{vid}$ \;
    }
    
    \mutmymethod{\free{}}{
        $largest\leftarrow Ord(\max(\hat{S}_{free}))+1$ \;
        \While{$u\leftarrow S_{sealed}.\spop{}$}{
            $\hat{S}_{free}$.\myinsert{$u$,\ $Ord(u)=largest$} \;
        }
    }
    
  \caption{\lrbu Cache}
  \label{alg:cache}
\end{algorithm*}
\end{figure}

\stitle{LRBU Cache.} 
We present our cache structure, \lrbu, short for \emph{least recent-batch used} cache. 
\refalg{cache} outlines the data structure of \lrbu, which consists of three members - $M_{cache}$, $\hat{S}_{free}$, and $S_{sealed}$. $M_{cache}$ stores the IDs of remote vertices as keys and their neighbours as values. $\hat{S}_{free}$ is an ordered set (\refsec{pre}) that keeps track of the orders of remote vertices that can be safely removed from the cache, where vertices with the smallest order can be replaced when the cache is full. $S_{sealed}$ represents a set of remote vertices that cannot be replaced at this time. 

There are 5 methods in \lrbu.
Given a vertex, \cget is used to obtain the neighbours if any and \ccontains checks whether the vertex presents in the cache (line~1-4). Unlike traditional cache structures, we design \cget and \ccontains to take only \emph{immutable} (i.e. read-only) references of the cache structure. As \cget and \ccontains are the two methods for reading the cache, such design makes cache read fully \textit{lock-free} when there is no concurrent writer. 

\cinsert is used to insert a remote vertex and its neighbours into the cache. Additionally, \creserve and \cfree are two unique methods of \lrbu. \creserve removes a vertex from $\hat{S}_{free}$ and adds it to $S_{sealed}$. 
\cfree pops all values in $S_{sealed}$ and adds them into $\hat{S}_{free}$. The released vertices will be given an order that is larger (line~12) than all existing vertices in $\hat{S}_{free}$. 
In the \cinsert method, replacement will be triggered if the cache is full. If $\hat{S}_{free}$ is not empty, the smallest vertex will be popped out for replacement. Thus, calling \creserve can prevent a particular vertex from being replaced when cache is full, while calling \cfree can make the certain vertices replaceable. If $\hat{S}_{free}$ is empty, the insertion will happen regardless of the capacity of the cache. This may cause the cache overflowed, but within only a limited amount as will be shown lately.

\stitle{Two-stage Execution Strategy.}
To make full use of \lrbu, we break down the execution of \extend into two separate stages, namely \emph{fetch} and \emph{intersect}. 
The algorithm of an \extend operator is given in \refalg{ext}.

\begin{figure}[t!]
\removelatexerror
\begin{algorithm*}[H]
    \small
    \setstretch{0.9}
    
    \SetKw{async}{async}
    \SetKw{parallel}{parallel}
    
    \SetKwFunction{fetch}{Fetch}
    \SetKwFunction{intersect}{Intersect}
    \SetKwProg{myproc}{Procedure}{}{}
    
    \SetKwFunction{getnbrs}{GetNbrs}
    
    \KwIn{Input channel $\mathbb{R}_i$, LRBU Cache $C$}
    \KwOut{Output channel $\mathbb{R}_{i+1}$}
    
    \myproc{\fetch{}}{
        $S_{remote} \leftarrow \{\}$ \;
        \parallel \ForAll{extended vertex $u \in \mathbb{R}_i$ }{
            $S_{remote}+=\{u\}$ \;
        }
        
        $S_{fetch}\leftarrow \{\}$ \;
        \ForEach{$u\in S_{remote}$}{
            \textbf{If}($C.$\contains{$u$}) 
            \textbf{then} $C.\reserve{u}$ 
            \textbf{else}   $S_{fetch}+=\{u\}$ \;
        }
        
        \async \ForEach{$(u,\mathcal{N}_G(u))\in\ $ \getnbrs{$S_{fetch}$}}{
            $C.$\myinsert{$u$,$\mathcal{N}_G(u)$} \;
        }
    }
    
    \myproc{\intersect{}}{
            $\mathbb{R}_{i+1}\leftarrow \{\}$ \;
            \parallel \ForAll{$p_i\in \mathbb{R}_i$}{
                $nbrs\_list\leftarrow \{\}$ \;
                \ForEach{extended vertex $u\in p_i$}{
                    \textbf{If} ($u$ is remote) 
                    \textbf{then}
                        $nbrs\_list+=\{C.\get{$u$}\}$ \;
                    \textbf{else} $nbrs\_list+=\{\mathcal{N}_G(u)\}$ \;
                }
                $\mathbb{C}(v_{i+1}) \leftarrow \cap_{nbrs\in nbrs\_list} nbrs$ \;
                \ForEach{$v\in \mathbb{C}(v_{i+1})$}{
                    \textbf{If}($v\not\in p_i$)
                    \textbf{then} $\mathbb{R}_{i+1}+=\{p_{i}+\{v\}\}$ \;
                }
        }
        
        $C$.\free{} \;
    }
    
    \Return{$\mathbb{R}_{i+1}$}\;
    
    \caption{Algorithm of \extend}
    \label{alg:ext}
\end{algorithm*}
\end{figure}

In the fetch stage, the \extend scans the input data and collects a set $S_{remote}$ of all remote vertices that need to be fetched in the current batch (line~2-4). It then checks for each remote vertex if the vertex is in the cache already (line~7). If the vertex has been cached, the extender seals this vertex in the cache, which prevents this particular entry to be replaced while processing this batch of data. Otherwise, it puts the vertex into a fetch set $S_{fetch}$. At the last step of the fetch stage, all vertices in $S_{fetch}$ will be fetched asynchronously by sending the \nbrs RPC in batches and inserted into the shared cache using one single writer (line~8-9). Note that cache write can be well overlapped with the asynchronous RPC requests. In the intersect stage, the extender performs the multiway intersections defined in \refeq{wco} to obtain the results and send them to the output (line~17). Finally, the sealed vertices are released by calling \cfree (line 20), which updates cache positions to allow them to be replaced thereafter.

In the execution, remote vertices are sealed at the beginning (line~7) and released at the end (line~20), which represents the vertices used in the very recent batch. As a result, even the cache is overflowed, the amount will not be more than the maximum number of the remote vertices in a batch. 
When the cache is full, \lrbu replaces the vertices with the smallest order, which must be the vertices from the least-recent batch (how \lrbu is named).



The two-stage execution strategy, together with the \lrbu cache structure, eventually leads to a zero-copy and lock-free cache access in \extend operator:

\sstitle{$\blacktriangleright$ Zero-copy.} Each vertex that will be accessed during the intersection is either in the local partition or sealed in the cache (line 15-16). As no modification will occur on the cache structure in the intersect stage (until next batch), we can access the vertex data by simply referencing the memory. 

\sstitle{$\blacktriangleright$ Lock-free.} Recall that the \get method of \lrbu is read-only and no write operation is executed during intersection. Cache write only happens in the stage of fetch (line 7 and 9), and at the end of extend (line 20). As we allow only one cache writer in each machine, the cache access (both read and write) in \page is completely lock-free. 



\begin{remark}
Our two-stage execution strategy separates fetch and intersect stages for lock-free and zero-copy cache access, which results in vastly improved performance. Synchronisation between fetch stage and intersect stage is necessary, but the overhead is very small as demonstrated in Exp-6 (\refsec{exp}).  
In addition, the initial scan in the fetch procedure can effectively aggregate RPC requests of pulling remote vertices, letting merged RPCs to be sent in bulk, which results in effective network utilisation.
\end{remark}
\section{Scheduling}
\label{sec:schedule}
We present in this section how we address the memory issue of subgraph enumeration by developing advanced scheduling techniques for \page. Note that there requires global synchronisation for \join operator to guarantee no missing results. To ease the presentation, we first assume that the execution plan contains no \join to focus on the two scheduling techniques - DFS/BFS-adaptive scheduling for bounded-memory execution and work stealing for load balancing. Then, we introduce how to plugin the \join operator. 

\subsection{Overview}
\page's scheduler is a daemon thread in each machine that maintains a shared scheduling channel with all workers. Each worker can either send its status to the scheduler or receive scheduling signals. Once an operator $O$ is scheduled (calling \texttt{schedule($O$)}), the scheduler will broadcast a \emph{schedule} signal to all workers to run $O$. The scheduler can also broadcast a \emph{yield} signal to yield the running of $O$ (calling \texttt{Yield($O$)}). The workers, once received the yield signal, will complete the current batch before going to sleep.


Without \join, the dataflow graph is a directed line graph. Thus, there will be at most one precursor (and successor) for each operator. Naively, there are two scheduling orders, depth-first-search (DFS) order and breadth-first-search (BFS) order. DFS scheduler will immediately yield the current operator and schedule the successor, as long as the current operator has completed \emph{one} batch of input data. When obtaining the final results from one batch, the scheduler backtracks to the starting point to consume the next batch. On the other hand, the BFS scheduler will sequentially schedule the operators in the dataflow and not move forward to the successor until it completes computing all input data batches. 


DFS scheduler may not fully utilize parallelism and network bandwidth\cite{BENU}, while BFS scheduler can suffer from memory crisis due to the maintenance of enormous intermediate results \cite{seed,twin-twig,edge-join,star-join}. Existing works use static heuristics such as region group \cite{RADS} and batching \cite{wco-join,patmat-exp} to constrain the scheduler to only consume a portion (e.g. a batch) of input data (vertice/edges) on the \scan operator, and will not move to the next portion until it sinks the final results of this portion. 
Such static heuristics lack in theoretical guarantee and can perform poorly in practice. We have observed out-of-memory errors from the experiment even while starting from one single vertex (e.g. on \dtcw in \refsec{exp}). 

\subsection{DFS/BFS-adaptive Scheduler}
\label{sec:batch}


\begin{figure}[t!]
\removelatexerror
\begin{algorithm*}[H]
    \small
    \setstretch{0.9}
    
    \SetKwFor{Loop}{Loop}{}{EndLoop}
    \SetKw{Break}{break}
    
    \SetKwFunction{run}{Schedule}
    \SetKwFunction{yield}{Yield}
    \SetKwFunction{precursor}{Precursor}
    \SetKwFunction{successor}{Successor}

    \KwIn{An execution plan $P$}
    
    $O\leftarrow$ first operator in $P$ \; 
    
    \While{there are uncompleted operators}{
        \If {$O$ has no input $\land O\neq \scan$}{
            $O\leftarrow P.$\precursor{$O$} \;
        }\Else{
            \run{$O$} \;
            
            \Loop{}{
                \If{$Q_{O}$.is\_full()  $\lor$ $O$ has no input}{
                    \yield{$O$} ;
                    \Break \;
                }
            }
            
            \textbf{If}($O = \sink$) 
            \textbf{then}  $O\leftarrow P.$\precursor{$O$} \;
            \textbf{else} $O\leftarrow P.$\successor{$O$} \;
           
        }
    }
  \caption{DFS/BFS-adaptive Scheduler}
  \label{alg:batching}
\end{algorithm*}
\end{figure}

We propose a \emph{DFS/BFS-adaptive scheduler} for \page to bound the memory usage while keeping high network and CPU utilisation. Specifically, we equip a fixed-capacity output queue $Q_{O}$ for each output channel of all operators in \page. \refalg{batching} illustrates the algorithm. Once an operator is scheduled, the scheduler tends to let it consume as much input data as possible to drive high the CPU utilisation. Meanwhile, each worker will report the number of results in $Q_O$ to the scheduler once it completes computing one batch of data. 
Whenever $Q_O$ is full, it broadcasts the ``yield'' signal to all workers to yield the current operator, preventing it from consuming any more batches (line~9).
The successor is then scheduled to consume the output of the current operator (line~11). If all results in the input channel are consumed, the scheduler backtracks to the precursor (line~4) and repeats the process until the data in all operators has been consumed. 
Backtracking is always triggered on \sink because it consumes all input data directly (line~10).

\begin{figure}[t]
    \centering
    \includegraphics[width=0.48\textwidth]{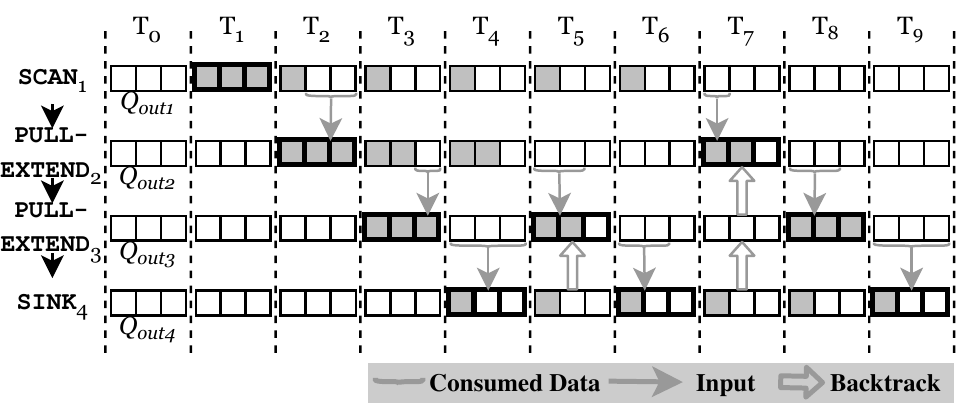}
    \caption{Running Example of DFS/BFS-adaptive Scheduler}
    \label{fig:batching}
\end{figure}


\begin{example}
An example is shown in \reffig{batching} (time slot $T_i$), with each block represents one batch of data and the operator under schedule highlighted.
Each operator has its own output queue $Q_{out_{i}}$ with fixed size equals to three batches. 
All queues are initially empty ($T_0$). The \scannospace$_1$ operator scans the data graph at $T_1$, outputting 3 batches. As the output queue is full, the scheduler yields the operator and schedules \extendnospace$_2$ at $T_2$. The process repeats until $T_4$, where the input of \sinknospace$_4$ becomes empty. Thus, the scheduler yields \sinknospace$_4$ and triggers backtracking. It moves to the precursor \extendnospace$_3$, and schedules this operator at $T_5$. Backtracking is also triggered at $T_6$ where the input of current operator becomes empty. However, when the scheduler backtracks to \extendnospace$_3$, its input is also empty. So the scheduler further moves forward to \extendnospace$_2$ and starts scheduling \extendnospace$_2$ at $T_7$.
\end{example}



\stitle{Bounded-Memory Execution.} Different from the static heuristics \cite{wco-join,RADS} that lack in a tight bound, we show how the DFS/BFS-adaptive scheduler helps bound memory consumption. Note that \sink operator directly writes data to the external devices and has no need of holding data, 
which is hence excluded from the memory analysis. We first present the following lemma for an \extend operator.

\begin{lemma}
\setstretch{0.9}
\label{lem:bound_extend}
The memory bound of scheduling a \extend operator is $O(|V_q|\cdot D_G)$.
\end{lemma}

\begin{proof}
For a \extend operator, we analysis the size of its output queue. Recall that the output queue has a fixed size, and \page's scheduler yields the operator when the output queue is full. However, as \page computes at least one batch of data (\refsec{opt}) at a time, the output queue can be potentially overflowed by the results of one batch of data. Given the size of a batch as $|batch|$, the maximum number of partial results that a single data batch can generate is $|batch|\cdot D_G$. Then, we need to consider the size of each partial result. Since \page stores each partial results as a compact array $\{u_1,u_2,\dots,u_{|v_q|}\}$, the size of each partial result is $O(|v_q|)$. Therefore, the memory bound is the product of $|batch|\cdot D_G$ and $O(|v_q|)$. As $|batch|$ is a pre-configured constant, we have the memory bound of scheduling a \extend operator is $O(|V_q|\cdot D_G)$.
\end{proof}

We discuss the other two cases in the following, namely \scan operator and the process of pulling-based hash join.

\sstitle{SCAN.} Note that the memory may overflow while enumerating a star (as the join unit). Thus, instead of directly computing the star, we rewrite a $\scan(q' = (v; \leaves))$ operator in a dataflow, via an initial $\scan(q_e = (v, v'))$ for any $v' \in \leaves$ to produce the first edge, which is then chained with $(|\leaves| - 1)$ \extend$(Ext = \{0\})$ operators to grow the other edges of the star. 



\sstitle{Pulling-based Hash Join.} Consider a join $(q', q'_l, q'_r)$ that is processed via pulling-based hash join, where $q'_r$ is a star $(v'_r; \leaves)$ (must be so according to \refeq{settings}). Similar to the \scan operator, a pulling-based hash join may also suffer from memory issue of computing stars. We show how such a join can be realized via a series of \extend operators to resolve the issue.

As a preliminary notation, given a query graph $q$ with the vertices listed as $\{v_1, v_2, \ldots, v_n\}$ and $V'_q \subseteq V_q$, we denote $\idx(q|V'_q)$ as an ordered indices of $q$ w.r.t. $V'_q$, where $i \in \idx(q|V'_q)$ if and only if $v_i \in V'_q$.
We split $\leaves$ into two parts, namely $V_1 = \leaves \cap V_{q'_l}$ and $V_2 = \leaves \setminus V_1$, and accordingly divide the execution into a chain of \extend operators.  Specifically,

\begin{itemize}[leftmargin=*]
    \item If $V_1 \neq \emptyset$, we deploy a \extend$(Ext = \idx(q'_l|V_1))$ operator. Note that this extension does not actually match new query vertex, but rather verify the connection between $v'_r$ and each $v \in V_1$ in a match. Thus, we install a hint on the operator to only preserve the result $f$ where $f(v'_r) = u_{i + 1}$, and get rid of the extended $u_{i+1}$ in the result.
    \item For each $v \in V_2$, we sequentially chain a new \extend$(Ext = \idx(q'_l|\{v\}))$ operator to grow the other star edges.
\end{itemize}

With the above transformations, we further have:
\begin{lemma}
\setstretch{0.9}
\label{lem:bound_others}
Given $q_s$ as a star $(v_s; \leaves)$, the memory bound of scheduling a \scan($q_s$) operator and a pulling-based hash join $(q', q'_l, q_s)$ are $O(|\leaves|^2\cdot D_G)$ and $O(|\leaves|\cdot|V_{q'}|\cdot D_G)$, respectively.
\end{lemma}



\begin{proof}
For \scan, we rewrite it into an initial scan and $|\leaves|-1$ \extend operators, and all of them are equipped with fixed-size output queues.
In the initial scan, each worker in the machine scans the local partition one vertex at a time. In the case of overflowing an output queue, the overflow is no more than the maximum number of edges that can be generated by one single vertex, which is $O(D_G)$.
There are $|\leaves|-1$ \extend operator followed by. By \reflem{bound_extend}, we know that the memory bound of each \extend operation is $O(|\leaves|\cdot D_G)$, so the total memory bound for \scan is $O(|\leaves|^2\cdot D_G)$.
\end{proof}

\begin{proof}
Similarly, for pulling-based hash join, it is divided into $|\leaves|$ \extend operations, where the memory bound of each \extend operator is $O(|V_q'|\cdot D_G)$. The overall memory bound of a pulling-based hash join is therefore $O(|\leaves|\cdot|V_{q'}|\cdot D_G)$.
\end{proof}

Summarizing from \reflem{bound_extend} and \reflem{bound_others}, we finally have:
\begin{theorem}
\setstretch{0.9}
\label{thm:memory_bound}
\page schedules a subgraph enumeration task with the memory bound of $O(|V_q|^2 \cdot D_G)$.
\end{theorem}

\begin{proof}
Consider a dataflow after \scan and pulling-based hash join are transformed to \extend operators. It contains at most $O(|V_q|)$ \extend operators, each of which consumes at most $O(|V_q| \cdot D_G)$ memory (\reflem{bound_extend}). Hence, the overall memory bound of \page to execute a a subgraph enumeration task is $O(|V_q|^2 \cdot D_G)$.
\end{proof}

\subsection{Load Balancing}
\label{sec:load_balance}
Graph computation is usually irregular due to the power-law characteristics in real-world graphs \cite{power-law-1,power-law-2}. 
Current solutions \cite{RADS,BENU} often distribute load based on the firstly matched vertex, which may still suffer from load skew. 
In \page, we adopt the work-stealing technique \cite{Fractal,gemini} to \emph{dynamically} balance the load. We implement a two-layer intra- and inter-machine work stealing to accommodate \page's caching mechanism and BFS/DFS-adaptive scheduler.


For intra-machine work stealing, we maintain a \emph{deque} \cite{Chase-Lev-deque} in each worker. Once the worker executes an operator, it injects the partial results $\mathbb{R}_i$ from the operator's input channel to its own deque. The current worker will pop out $\mathbb{R}_i$ from the back of the deque to do computation. Once a worker has completed its own job by emptying its deque, it will randomly pick one of the workers with non-empty deque, and steal half of the data from the front. For \extend operator, recall that its execution is separated into fetch and intersect stages. While there is barely any skew for fetching data, we only apply intra-machine work stealing to the intersect stage. Specifically, when a worker completes its computation in line~21 of \refalg{ext}, it will try to steal the other worker's unprocessed data in line~12 to continue the process.




Inter-machine work stealing happens when any machine completes computing its own job. In this case, the scheduler of the machine will send the \steal RPC to a random remote machine to steal unprocessed partial results in batches from the input channel of the \emph{top-most unfinished operator}. If receiving data, the scheduler will schedule the corresponding operator to compute the received data; otherwise, it picks another random machine to repeat the attempt. Machines who have completed their own work will send their status to the first machine in the cluster upon completion. The first machine will then broadcast the messages to all other machines in the cluster. A list of finished machines is maintained at each machine, whose job will not be stolen. Once the computation of stolen work is done and there is no more remote work to steal (i.e. all machines have finished their own job), the machine sends the status to the first machine again to mark termination.

Note that the work stealing is applied at operator-level as described to better balance the load. This is because the exponential nature of subgraph enumeration that can cause the intermediate results to explode at any operator on certain vertices (especially large-degree vertices).

\subsection{Handling Join Operator}
\label{sec:schedule_join}
\page enforces a synchronisation barrier prior to the \join operator, thus the join cannot proceed until both precursors complete their computation. With \join operator, the dataflow graph of \page becomes a directed tree. 

We first consider a dataflow $P$ with one \join operator (e.g. \reffig{5-path-dataflow}), which contains a left subgraph $P_1$ and a right subgraph $P_2$. \page first computes $P_1$, and then $P_2$, whose results will be globally synchronized at the barrier of \join. As $P_1$ and $P_2$ contains only \extend, they can be scheduled via the above scheduling techniques (\Cref{sec:batch,sec:load_balance}). 
\page computes the join after the computation of $P_1$ and $P_2$ are completed. 

Given $P_1$ and $P_2$, we use $P_1\dotarrow P_2$ to denote $P_1$ must be computed before $P_2$. In \reffig{5-path-dataflow}, we have $P_l \dotarrow P$ and $P_r \dotarrow P$. 
Each subgraph contains no \join can be directly scheduled; otherwise, it will be recursively divided by \join. By constructing a DAG of all subgraphs based on the $\dotarrow$ relations, a valid execution order can be determined via topological ordering of the DAG.

BFS/DFS-adaptive scheduling is unnecessary for \join, as the buffering technique (\refsec{join}) can already prevent memory from overflowing. While join may produce too many data to overflow the successors, we allow \join to actively tell the scheduler to yield its execution in case that its output queue is full. 
Regarding work stealing, we only apply intra-machine stealing for \join. For the non-trivial case that the buffered data is on disk, a worker can steal job by simply advancing the reading offsets of the other worker's buffered files.

\section{Applications}
\label{sec:applications}
\page is designed to be flexible for extending more functionalities. 
Extended systems can directly benefit from \page's pushing/pulling-hybrid communication and bounded-memory execution. We introduce three representative examples.

\sstitle{Cypher-based Distributed Graph Databases.} Subgraph enumeration is key to querying graph databases using language like Cypher \cite{cypher}. \page can thus be extended as a Cypher-based distributed graph database, by implementing more operations like projection, aggregation and ordering, and connecting it with a front-end parser (e.g. \cite{patmat-demo}) and an optimizer with cost estimation for labelled (and/or property) data graph (e.g. \cite{graphflow}). 

\sstitle{Graph Pattern Mining (GPM) Systems.} A GPM system \cite{Arabesque,AutoMine,Fractal,Peregrine} aims to find all subgraph patterns of interest in a large data graph. It supports applications such as motif counting \cite{motif} and frequent subgraph mining \cite{fsm}.
It essentially processes subgraph enumeration repeatedly from small query graphs to larger ones, each time adding one more query vertex/edge. 
Thus, \page can be deployed as a GPM system by adding the control flow like loop in order to construct a more complex dataflow for GPM tasks.

\sstitle{Shortest Path \& Hop-constrained Path.} \page can also be applied to solve more general path queries, such as the classic shortest path problem or hop-constrained path enumeration \cite{st-path}. Shortest path can be computed by repeatedly applying \extend from the source vertex until it arrives at the target. For hop-constrained path enumeration, \page can conduct a bi-directional BFS by extending from both ends and joining in the middle.
\section{Experiments}
\label{sec:exp}

\subsection{Experimental Setup}
We follow \cite{patmat-exp} to build a Rust codebase for a fair comparison. 
For join-based algorithms (\bigjoin and \seed), we directly adopt the Rust implementations in \cite{patmat-exp}, which contains many optimisations (e.g. symmetry break and compression). For \rads, the original authors have reviewed our implementation. 
For \benu, we select the distributed key-value database Cassandra \cite{cassandra} to store the data graph as recommended by the original authors. For others, we partition and store the data graph in the compressed sparse row (CSR) format and keep them in-memory.
We use the generic compression optimisation \cite{crystaljoin} whenever it is possible in all implementations, and decompress (by counting) to verify the results. 

\stitle{Hardware.} We deploy \page in: (1) a local cluster of 10 machines, each with a 4-core Intel Xeon CPU E3-1220, 64GB memory, 1TB disk, connected via a 10Gbps network; (2) an AWS cluster of 16 ``r5.8xlarge'' instances, each with 32 vCPUs, 256GB memory, 1TB Amazon EBS storage, connected via a 10Gbps network. We run 4 workers in the local cluster and 14 workers in the AWS cluster. All experiments are conducted in the local cluster except Exp-3. 

\stitle{Datasets.} We use 7 real-world datasets of different sizes in our experiments as in \reftable{data}. \textsf{Google} (\dtgo), \textsf{LiveJounal} (\dtlj), \textsf{Orkut} (\dtor), and \textsf{Friendster} (\dtfs) are downloaded from \cite{snap}. \textsf{UK02} (\dtuk), \textsf{EU-road} (\dteu), and \textsf{ClueWeb12} (\dtcw) are obtained from \cite{webgraph}, \cite{challenge9}, and \cite{clubweb}, respectively. The datasets include social graphs (\dtlj, \dtor and \dtfs), road networks (\dteu), and web graphs (\dtgo, \dtuk and \dtcw). 

\stitle{Queries.} We use 7 queries according to prior works \cite{seed,crystaljoin,wco-join,patmat-exp,BENU,RADS} as shown in \reffig{queries}. The partial orders for symmetry breaking are listed below each query.

\stitle{Parameters and Metrics.} If not otherwise specified, we use $q_1$-$q_3$ as the default queries, and \dtuk as the default dataset. Note that we may omit certain results for clarity. 
We configure the default system parameters of \page as batch size: $512K$ (\refsec{arch}), cache capacity: $30\%$ of the data graph (\refsec{pull}), and output queue size: $5{\times}10^7$ (\refsec{batch}). 
We allow 3 hours for each query. \timeout and \oom are used to indicate a query runs overtime and out of memory, respectively. We measure the total time $T$, computation time $T_R$ and communication time $T_C = T - T_R$ according to \cite{patmat-exp}. In the bar char, we present the \emph{ratio} of $\frac{T_C}{T}$ using grey filling, and mark the case of \oom with a \bm{$\mathbin{\textcolor{red}{\times}}$} on top of the bar.

\begin{table}[t]
\caption{Table of Datasets}
\label{tab:data}
\small
\setstretch{0.9}

\begin{tabular}{lrrrr}
\hline
Dataset & $|V|$       & $|E|$          & $d_{max}$  & $d_{avg}$ \\ \hline
\dtgo   & 875,713     & 4,322,051      & 6,332      & 5.0       \\
\dtlj   & 4,847,571   & 43,369,619     & 20,333     & 17.9      \\
\dtor   & 3,072,441   & 117,185,083    & 33,313     & 38.1      \\
\dtuk   & 18,520,486  & 298,113,762    & 194,955    & 16.1      \\
\dteu   & 173,789,185 & 347,997,111    & 20         & 3.9       \\
\dtfs   & 65,608,366  & 1,806,067,135  & 5,214      & 27.5      \\
\dtcw   & 978,409,098 & 42,574,107,469 & 75,611,696 & 43.5      \\ \hline
\end{tabular}
\end{table}

\begin{figure}[t]
  \centering
  \includegraphics[width=.4\textwidth]{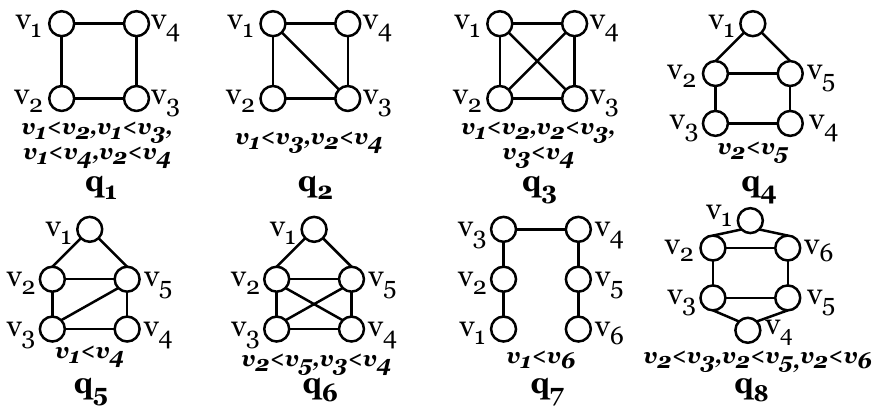}  
  \caption{The Query Graphs}
  \label{fig:queries}
\end{figure}

\subsection{Comparing Existing Solutions}

\begin{figure*}[ht]
    \centering
    \begin{subfigure}[b]{.5\textwidth}
        \centering
        \includegraphics[height = 0.08in]{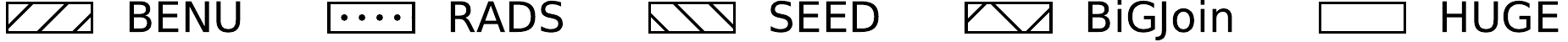}
    \end{subfigure} %
    \\
    \begin{minipage}{.33\textwidth}
        \centering
        \begin{subfigure}{.4\textwidth}
          \centering
          \includegraphics[width=\linewidth]{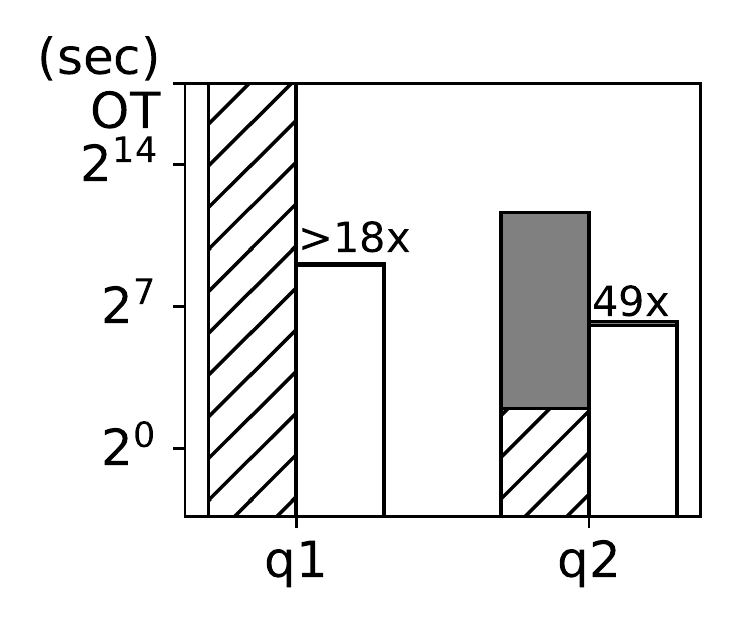}  
          \caption{\pagebenu}
        \end{subfigure} %
        ~
        \begin{subfigure}{.4\textwidth}
          \centering
          \includegraphics[width=\linewidth]{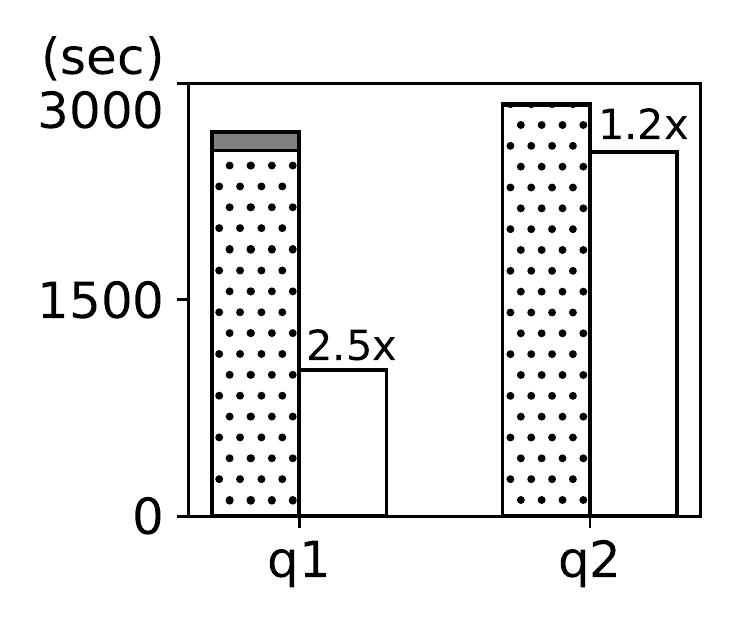}  
          \caption{\pagerads}
        \end{subfigure} %
        \\
        \begin{subfigure}{.4\textwidth}
          \centering
          \includegraphics[width=\linewidth]{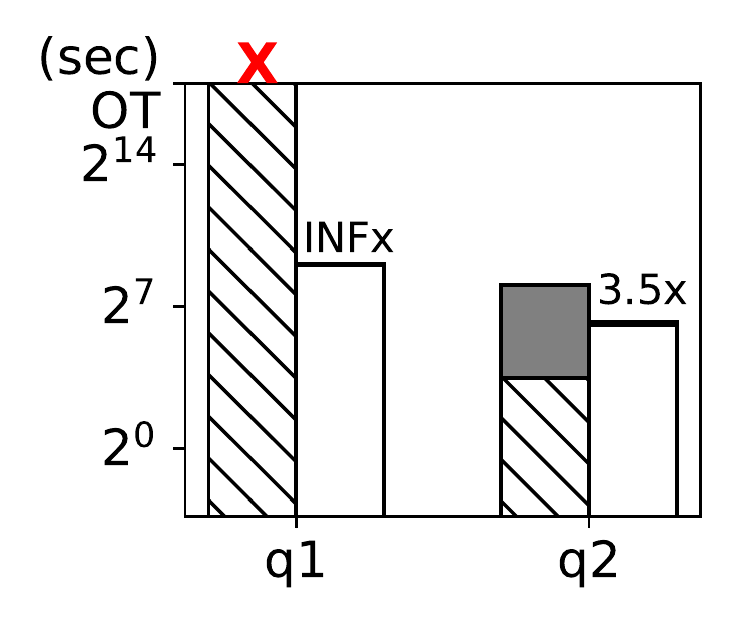}  
          \caption{\pageseed}
        \end{subfigure} %
        ~
        \begin{subfigure}{.4\textwidth}
          \centering
          \includegraphics[width=\linewidth]{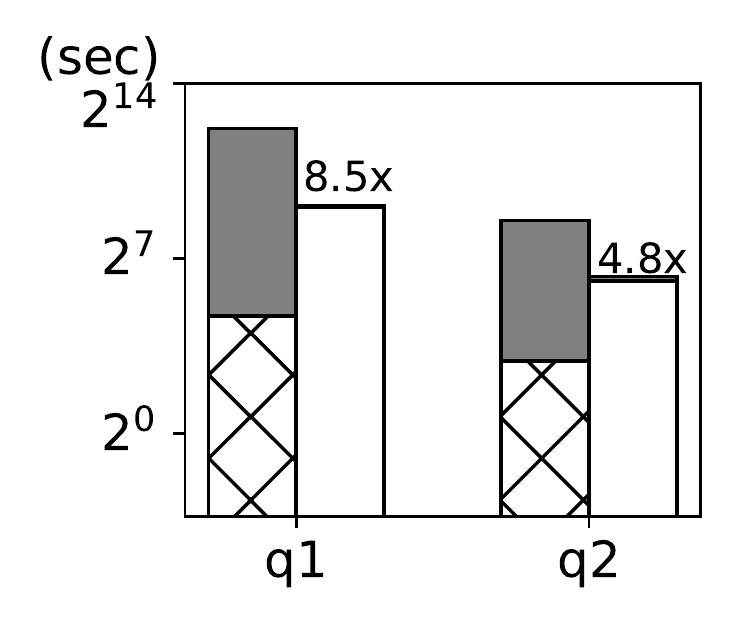}  
          \caption{\pagewco}
        \end{subfigure} %
        \caption{Speed Up Existing Algorithms}
        \label{fig:speed-up}
    \end{minipage} %
    ~
    \begin{minipage}{.67\textwidth}
        \begin{subfigure}{0.33\textwidth}
          \centering
          \includegraphics[width=\linewidth]{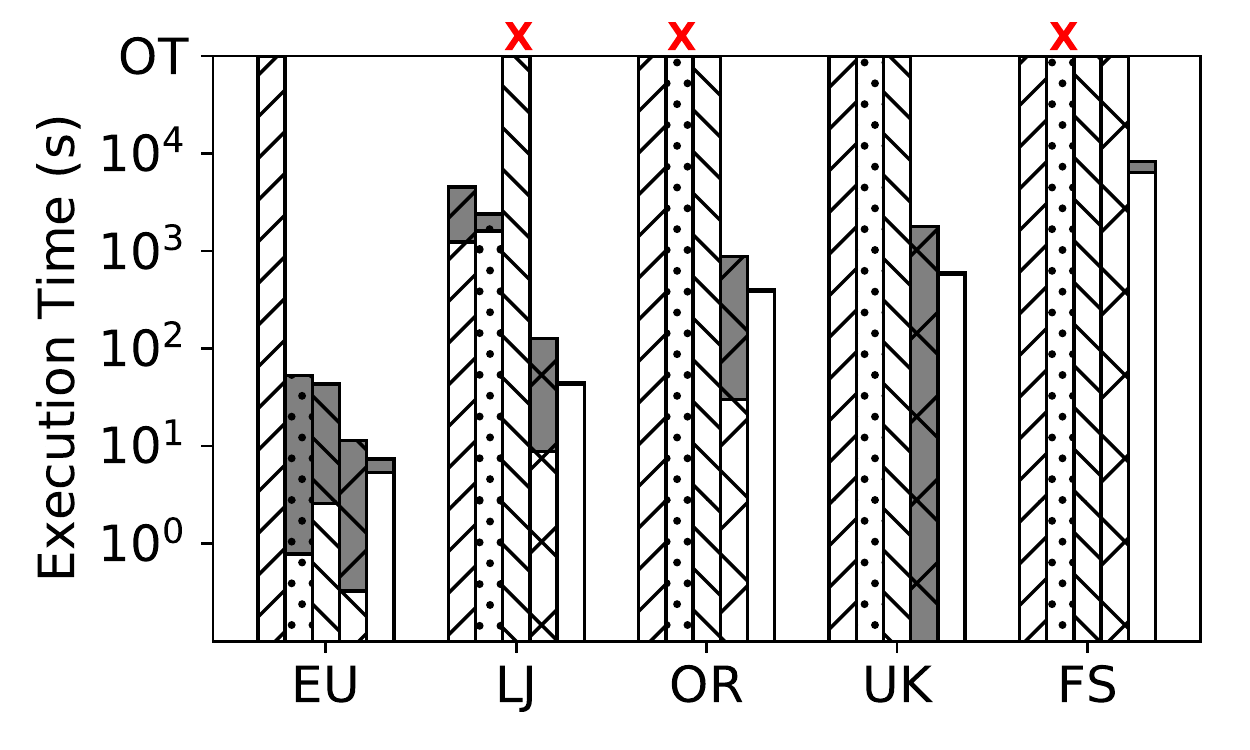}  
          \caption{$q_1$}
        \end{subfigure} %
        ~
        \begin{subfigure}{0.33\textwidth}
          \centering
          \includegraphics[width=\linewidth]{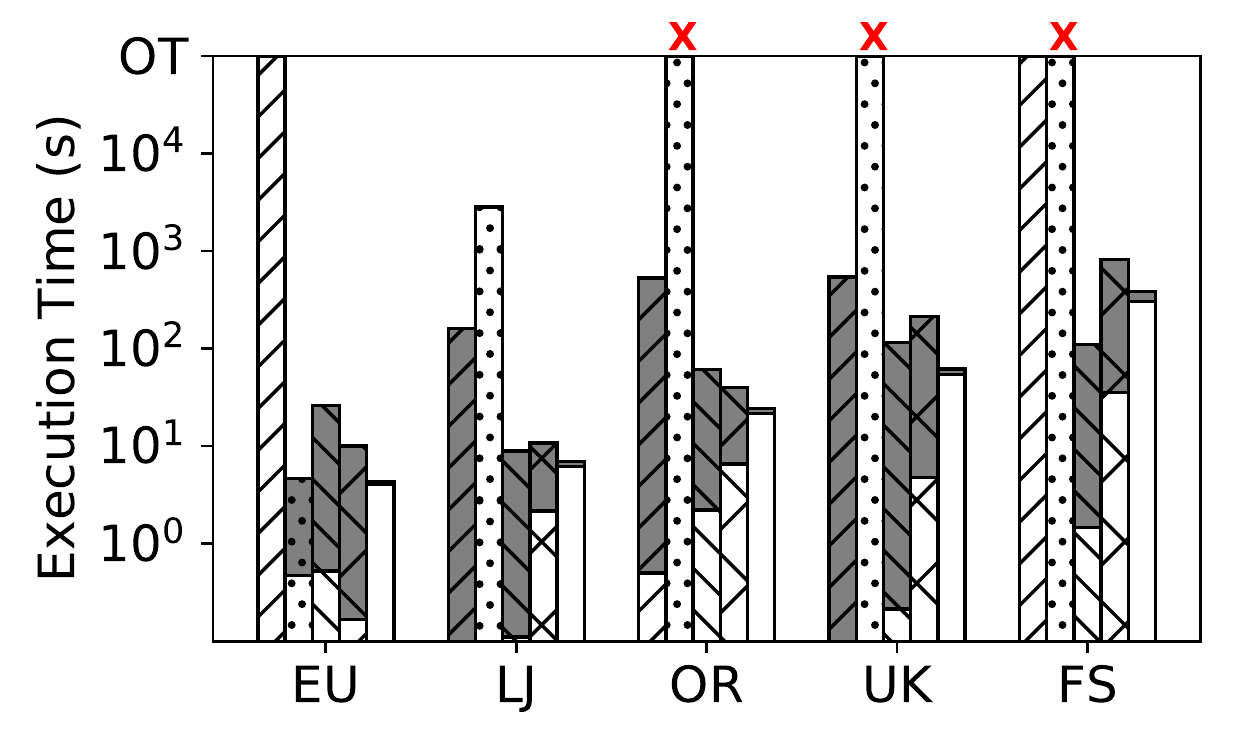}  
          \caption{$q_2$}
        \end{subfigure} %
        ~
        \begin{subfigure}{0.33\textwidth}
          \centering
          \includegraphics[width=\linewidth]{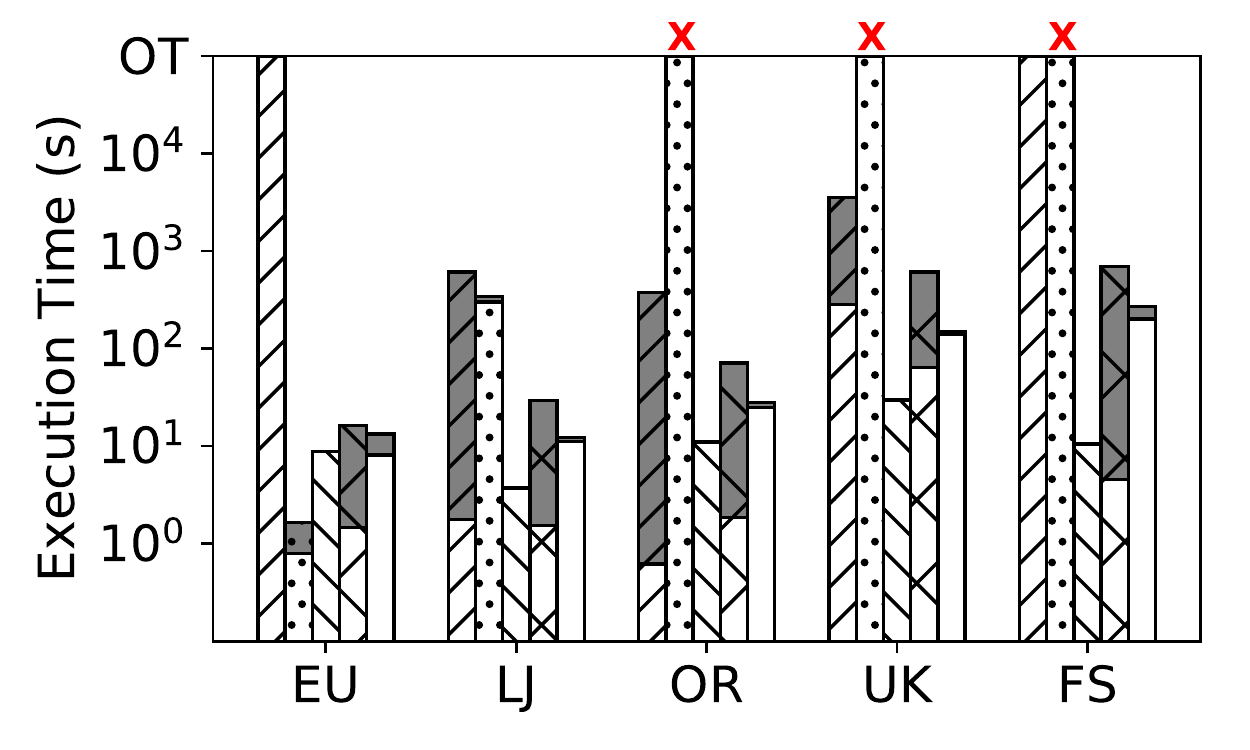}  
          \caption{$q_3$}
        \end{subfigure} %
        \\
        \begin{subfigure}{0.33\textwidth}
          \centering
          \includegraphics[width=\linewidth]{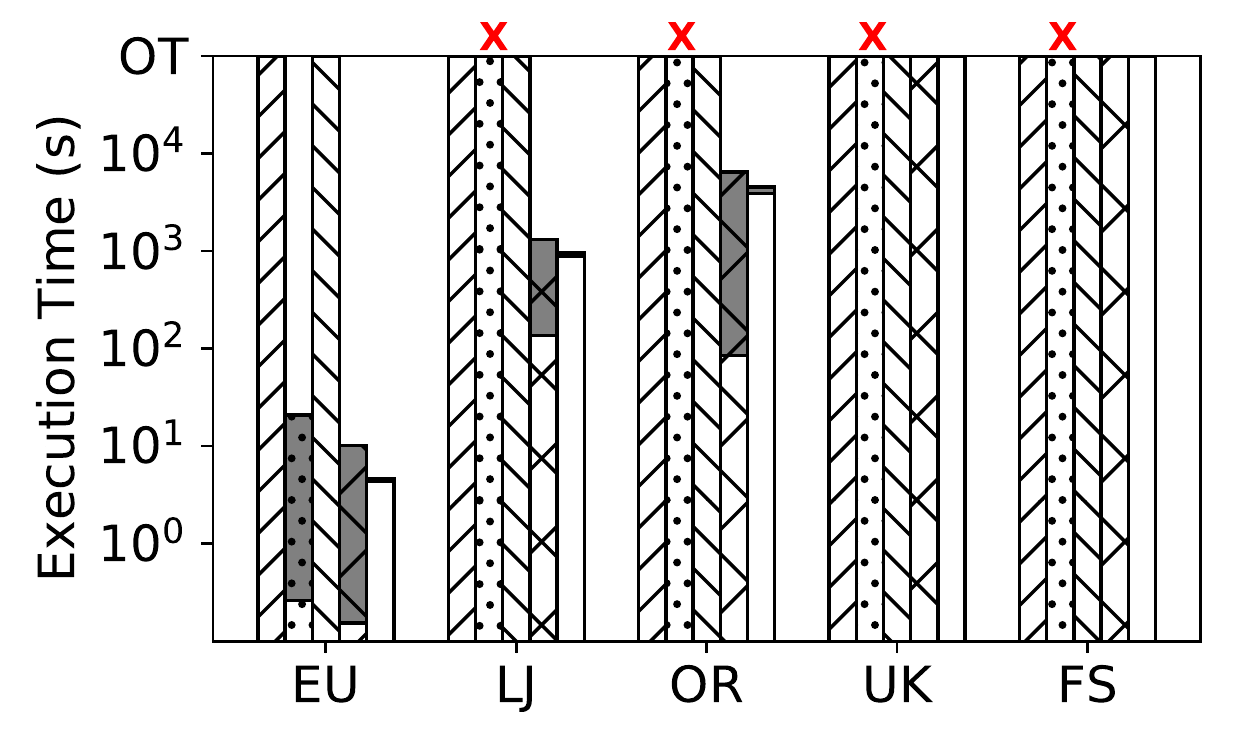}  
          \caption{$q_4$}
        \end{subfigure} %
        ~
        \begin{subfigure}{0.33\textwidth}
          \centering
          \includegraphics[width=\linewidth]{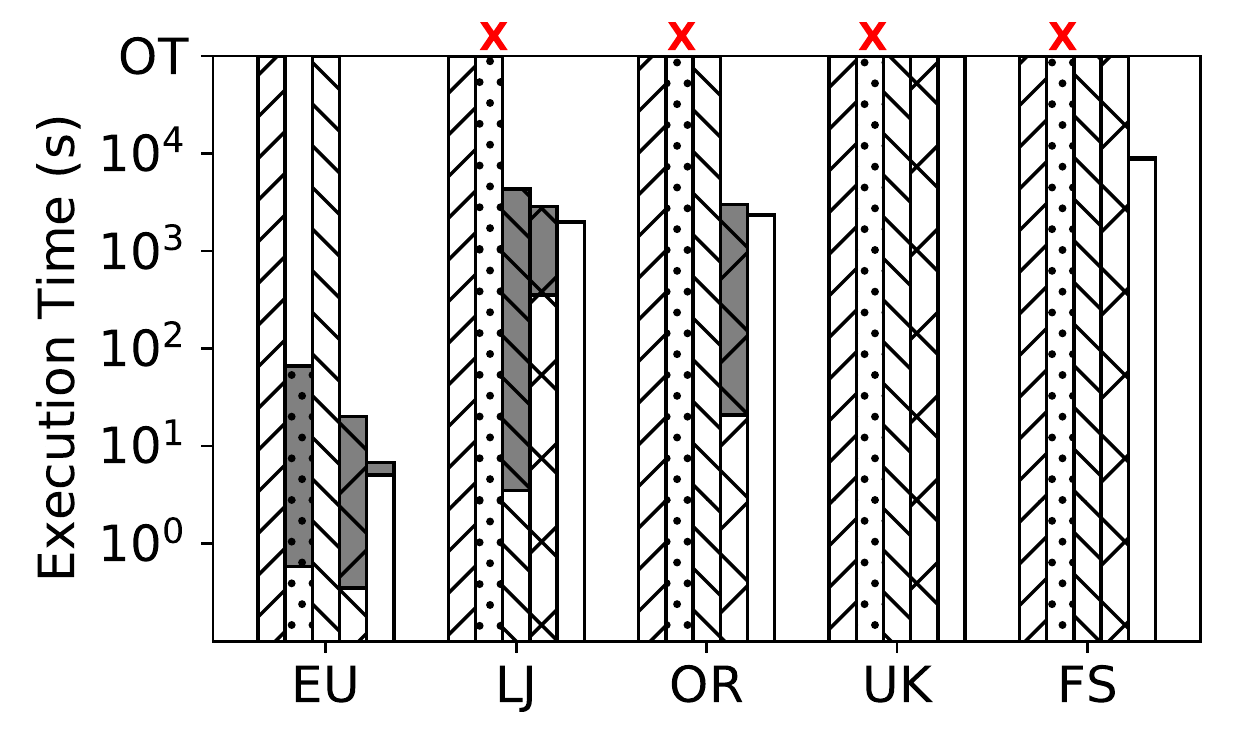}  
          \caption{$q_5$}
        \end{subfigure} %
        ~
        \begin{subfigure}{0.33\textwidth}
          \centering
          \includegraphics[width=\linewidth]{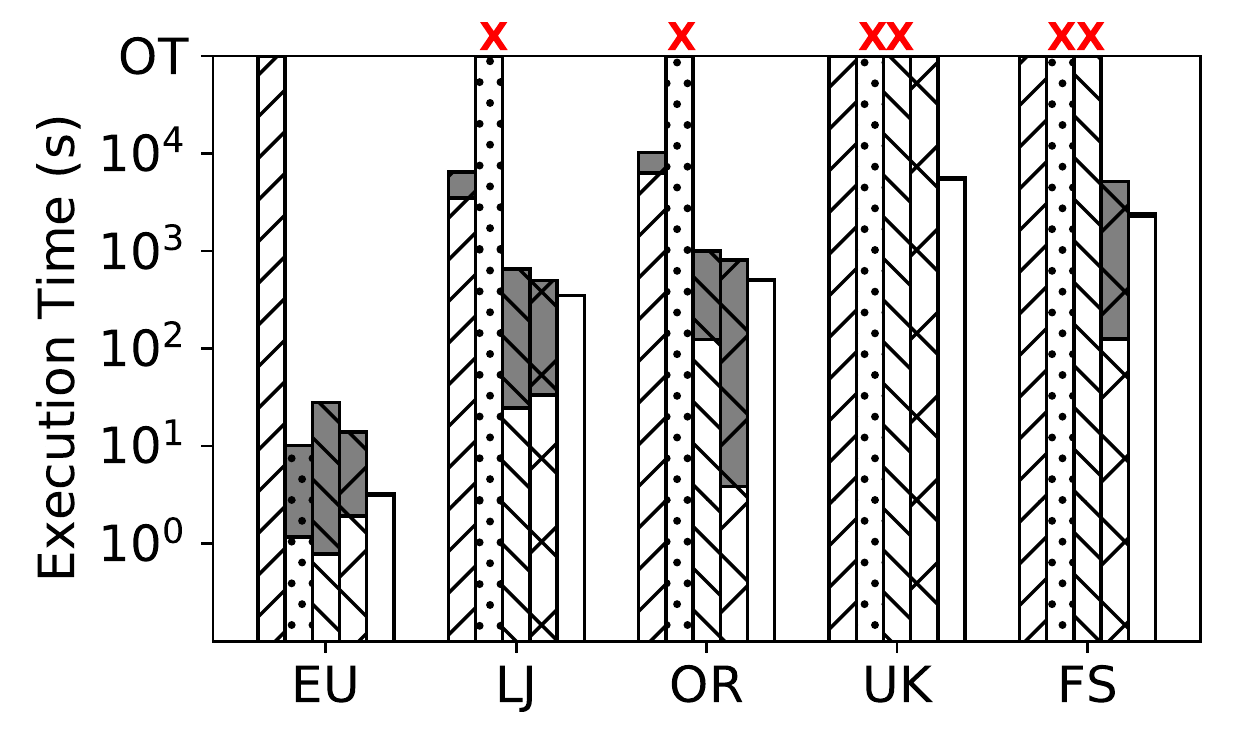}  
          \caption{$q_6$}
        \end{subfigure} %

        \caption{All-Round Comparisons.}
        \label{fig:all-round}
    \end{minipage}
\end{figure*}

\stitle{Exp-1: Speed Up Existing Algorithms.} We first verify that existing works can be readily plugged into \page via their logical plans to receive automatic speedup and bounded-memory execution (\refrem{plug-in-exist}). We run the logical plans of \benu, \rads, \seed, and \bigjoin in \page, denoted as \pagebenu, \pagerads, \pageseed, and \pagewco, respectively. While \seed's plan may include clique as the join unit, we let \pageseed compute the clique via \extend instead of building the costly triangle index. Note that we use \dtlj instead of \dtuk for \rads and \pagerads, where both of them run \timeout on \dtuk because of the poor execution plan of \rads. The results of $q_1$ and $q_2$ are presented in \reffig{speed-up}, with the speedup factor marked on top of each bar.

For \benu, the huge overhead of using Cassandra makes it significantly slower than \pagebenu. 
For \rads, 
the speedup is less significant, mainly due to the poor execution plans of \rads, especially for $q_2$, where a massive number of 3-stars must be materialized. 
\seed runs \oom for $q_1$, while \pageseed completes in 544 seconds because it processes the join via the more efficient pulling-based \wopt join according to \refeq{settings}.
Note that although \seed replies on the triangle index for querying $q_2$, our index-free \pageseed still achieves a speedup of 2.5$\times$. Lastly, \pagewco outperforms \bigjoin by 8.5$\times$ and 4.8$\times$ on $q_1$ and $q_2$, with less memory usage (e.g. 4GB vs 12GB for $q_1$). Specifically, \pagewco reduces the communication time by 764$\times$ and 115$\times$, respectively, thanks to the efficient \extend operator.  

\stitle{Exp-2: All-round Comparisons.} We compare \page (with optimal execution plan by \refalg{dp-opt}) on $q_1$-$q_6$ with the state-of-the-art algorithms using different data graphs in this experiment (\reffig{all-round}). Among all cases, \page has the highest completion rate of 90\%, where \bigjoin, \seed, \rads, and \benu complete 80\%, 50\%, 30\%, and 30\%, respectively.
Computation-wise, \page outperforms \rads by 54.8$\times$, \benu by 53.3$\times$, \seed by 5.1$\times$, and \bigjoin by 4.0$\times$ on average. Note that with the costly triangle index, \seed can query $q_3$ (a clique) without any join, while the index-free \page only runs slightly slower for this query. 
Communication-wise, the communication time of \page takes only a very small portion (the shaded area in a bar) in all cases, due to a combination of caching, batching RPC requests, and good execution plan. In comparison, we can observe that all other algorithms (especially join-based algorithms) spend a notable portion of time communicating data in most cases. 
Memory-wise, due to the BFS/DFS-adaptive scheduling technique, \page keeps the memory usage bounded, and the peak memory usage is 16.6GB among all cases, compared to $>$64GB (\oom), 2.3GB, $>$64GB, 34.1GB for \rads, \benu, \seed and \bigjoin, respectively. This experiment shows that \page can perform scalable and efficient subgraph enumeration while taking into consideration of computation, communication and memory management.

\begin{table}[h]
\small
\setstretch{0.9}
\caption{Throughput on \dtcw}
\begin{tabular}{|c|c|c|c|}
\hline
& $q_1$ & $q_2$ & $q_3$ \\ \hline \hline
Throughput & 2,895,179,286/s  &  354,507,087,789/s  &  206,696,071/s   \\ \hline
\end{tabular}
\label{tab:web-scale}
\end{table}

\stitle{Exp-3: Web-scale Data Graph.} We run \page over the web-scale graph \dtcw on the AWS cluster to test its ability in handling large graphs. The data graph has a raw size of about 370GB (in CSR format) which is larger than the configured memory of the machine. \benu fails to load the graph into Canssandra within one day, so as \seed that needs to build the triangle index. Both \rads and \bigjoin run \oom quickly even when we start with one single vertex in a region group (batch). However, \page, runs the queries with a stable memory usage of around 85G when setting the cache capacity size to 30GB and the output queue size to $5\times10^8$. The number of results on this graph has been estimated to be dramatically large \cite{wco-join}. Therefore, we run each query for 1 hour and report the average throughput ($\frac{|\mathbb{R}|}{3600}$) of \page in \reftable{web-scale}. The authors of \bigjoin \cite{wco-join} have used an incremental dataflow to avoid overflowing the memory. In a same-scale cluster (the machine has similar configurations), they obtain the throughput of 26,681,430/s and 46,517,875/s for $q_1$ and $q_3$, which is much lower than our results. 

\subsection{The Design of \page}

\begin{figure}[h]
    \centering
    \begin{subfigure}[b]{.5\textwidth}
    \centering
    \includegraphics[height = 0.08in]{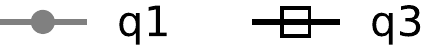}
    \end{subfigure}%
    \\
    \begin{subfigure}{.166\textwidth}
      \centering
      \includegraphics[width=\linewidth]{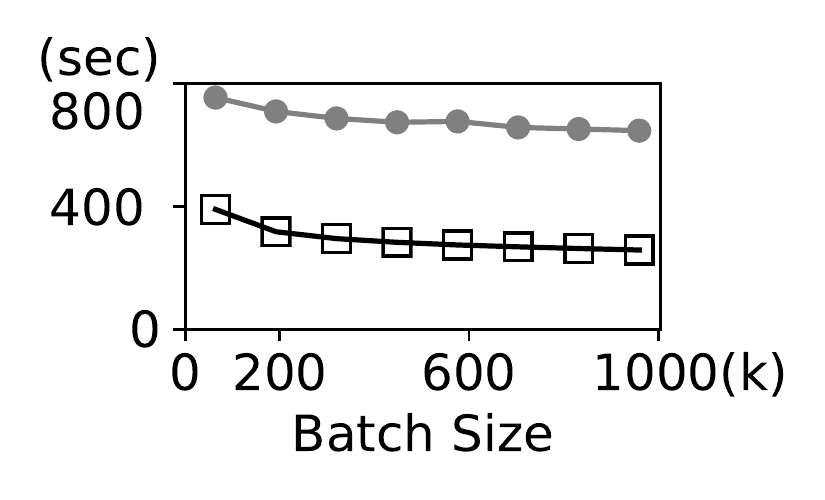}  
      \vspace{-2\baselineskip}
      \caption{Execution Time}
    \end{subfigure} %
    ~
    \begin{subfigure}{.166\textwidth}
      \centering
      \includegraphics[width=\linewidth]{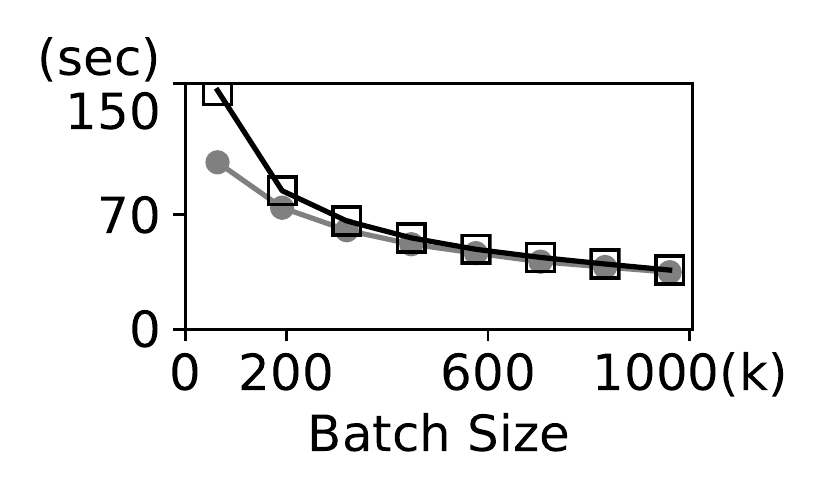}  
      \vspace{-2\baselineskip}
      \caption{Communication Time}
    \end{subfigure} %
    ~
    \begin{subfigure}{.166\textwidth}
      \centering
      \includegraphics[width=\linewidth]{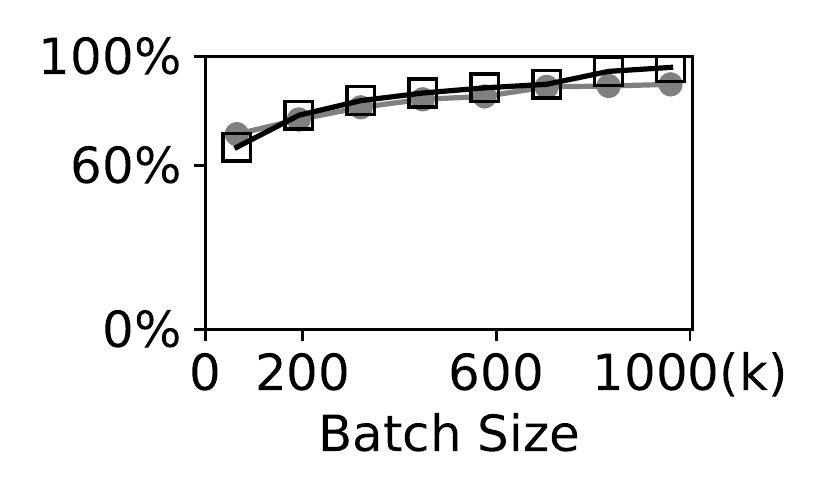}  
      \vspace{-2\baselineskip}
      \caption{Network utilisation}
    \end{subfigure} %
    
    \caption{Vary Batch Size}
    \label{fig:exp-batch}
    \vspace*{-.8\baselineskip}
\end{figure}

\stitle{Exp-4: Effectiveness of Batching.} We use a batch of data as the minimum data processing unit (\refsec{arch}). We investigate how batching affects the \page's performance by varying the batch sizes with cache disabled. We report the results 
of $q_1$ and $q_3$ in \reffig{exp-batch}. Let the size of data transferred via network be $C$ (in GB). We measure the network utilisation as $\frac{8C / T_C}{10}$ (10Gbps is the network bandwidth).
Increasing the batch size reduces execution and communication time. This is because \page's two-stage execution strategy can efficiently aggregate RPC requests within a single batch to improve network utilisation. The average network utilisation starts with $71\%$ when the batch size is 100K, and arrives at $86\%$ and $94\%$ when the size is 512K and 1024K, respectively. As larger batch can make the cache and the output queue more easily overflowed, we set the default batch size to 512K (with satisfactory network utilisation).


\begin{figure}[h]
    \centering
    \begin{subfigure}[b]{.5\textwidth}
        \centering
        \includegraphics[height = 0.08in]{figures/exp/batch-legend.pdf}
    \end{subfigure}%
    \\
    \begin{subfigure}{.166\textwidth}
      \centering
      \includegraphics[width=\linewidth]{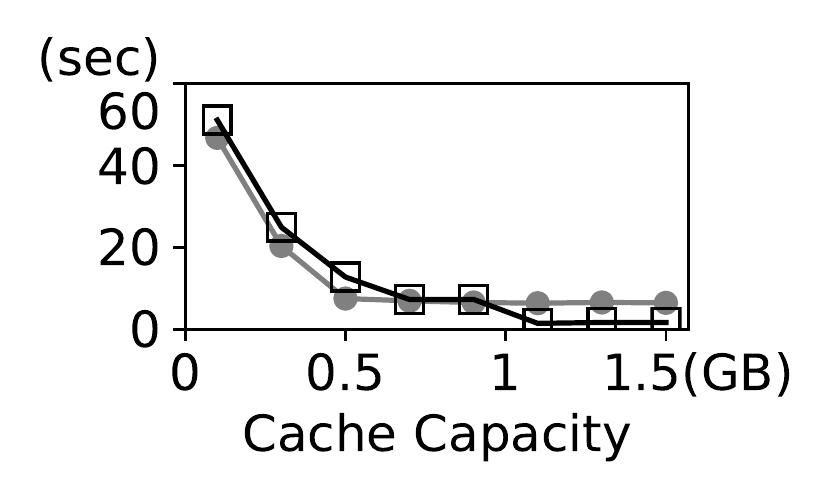}  
      \vspace{-2\baselineskip}
      \caption{Communication Time}
    \end{subfigure} %
    ~
    \begin{subfigure}{.166\textwidth}
      \centering
      \includegraphics[width=\linewidth]{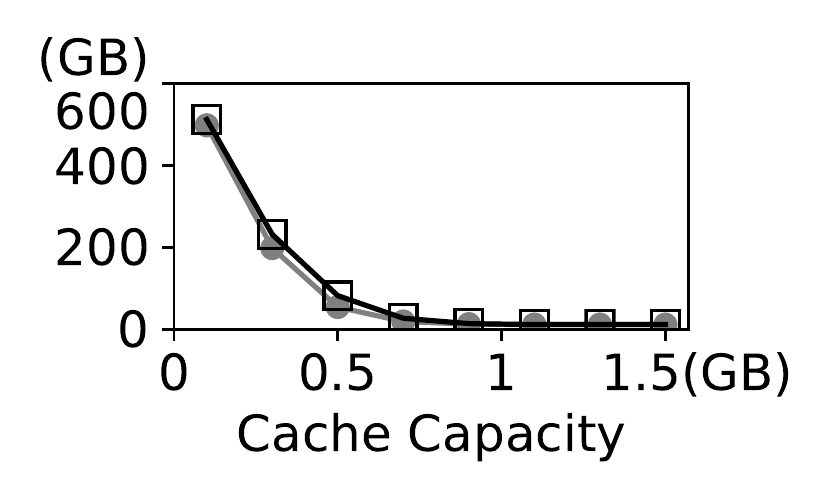}  
      \vspace{-2\baselineskip}
      \caption{Communication Size}
    \end{subfigure} %
    ~
    \begin{subfigure}{.166\textwidth}
      \centering
      \includegraphics[width=\linewidth]{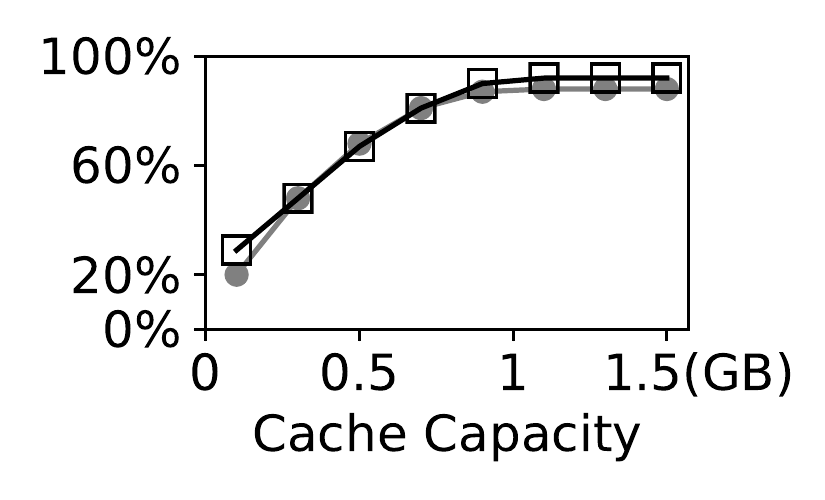}  
      \vspace{-1.8\baselineskip}
      \caption{Hit Rate}
    \end{subfigure} %
    
    \caption{Vary Cache Capacity}
    \label{fig:cache-cap}
\end{figure}

\stitle{Exp-5: Cache Capacity.} We evaluate the impacts of cache capacity on query performance in \reffig{cache-cap}, varying the cache capacity from 0.1GB to 1.5GB. 
As the capacity increases, the communication time and size decrease rapidly. For $q_1$, growing the cache capacity from 0.1GB to 0.5GB increases the average hit rate by about 3.5 times, and reduces the total communication by almost 10 times. The performance does not change after the cache capacity exceeds 1.1GB 
for both queries, whose capacity is enough to hold all remote vertices accessed in these two queries. 


\begin{table}[h]
\small
\setstretch{0.9}

\caption{Vary Cache Design}
\begin{tabular}{|c|c|c|c|c|c|}
\hline
      & \lrbu            & \lrbu-Copy  & \lrbu-Lock  & \lru-Inf & Cncr-\lru \\ \hline \hline
$q_1$ & \textbf{589.3s} (27.7s)  & 734.1s    & 920.1s  & 997.5s   & 2597.1s  \\ \hline
$q_2$ & \textbf{63.3s} (3.7s)   & 74.5s     & 98.0s   & 107.7s   & 240.5s    \\ \hline
$q_3$ & \textbf{200.6s} (24.8s)  & 314.5s    & 525.4s  & 563.4s   & 980.9s   \\ \hline
\end{tabular}
\label{tab:cache-design}
\end{table}

\stitle{Exp-6: Cache Design} 
We evaluate the benefit of \lrbu's lock-free and zero-copy cache design (\reftable{cache-design}). While enabling the two-stage execution strategy (\refsec{pull}), we first compare the performance of \page with \lrbu, \lrbu-Copy, \lrbu-Lock and \lru-Inf, which represent our \lrbu cache, the \lrbu cache with memory copy enforced, the \lrbu cache with both memory copy and lock enforced, and a \lru cache with infinite capacity\footnote{The official Rust LRU library (\url{https://doc.rust-lang.org/0.12.0/std/collections/lru_cache}) is used, with the capacity set to the maximum integer.}, respectively. 
\lrbu outperforms \lrbu-Copy, \lrbu-Lock and \lru-Inf by 1.3$\times$, 1.9$\times$ and 2.0$\times$, respectively, which reveals the effectiveness of the zero-copy and lock-free cache access. 
To evaluate the two-stage execution strategy, we further compare \lrbu with a variant called Cncr-\lru, which disables two-stage execution, and applies advanced concurrent \lru-cache \cite{cache_io,caffeine} for resolving data contentions. The performance gain of \lrbu over Cncr-\lru is 4.4$\times$ on average. The two-stage execution may bring in synchronisation cost, which is hard to measure directly. Alternatively, as indicated in the bracket of \lrbu in \reftable{cache-design}, we measure the whole time spent on the fetch stage $t_f$, knowing that it must contain the time for synchronisation. Observe that $t_f$ is merely about 7.5\% of the total execution time on average, the synchronisation overhead must thus be small.



\begin{figure}[h]
    \centering
      \includegraphics[width=0.4\linewidth]{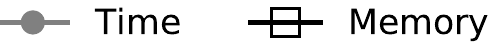} \\
      \includegraphics[width=0.5\linewidth]{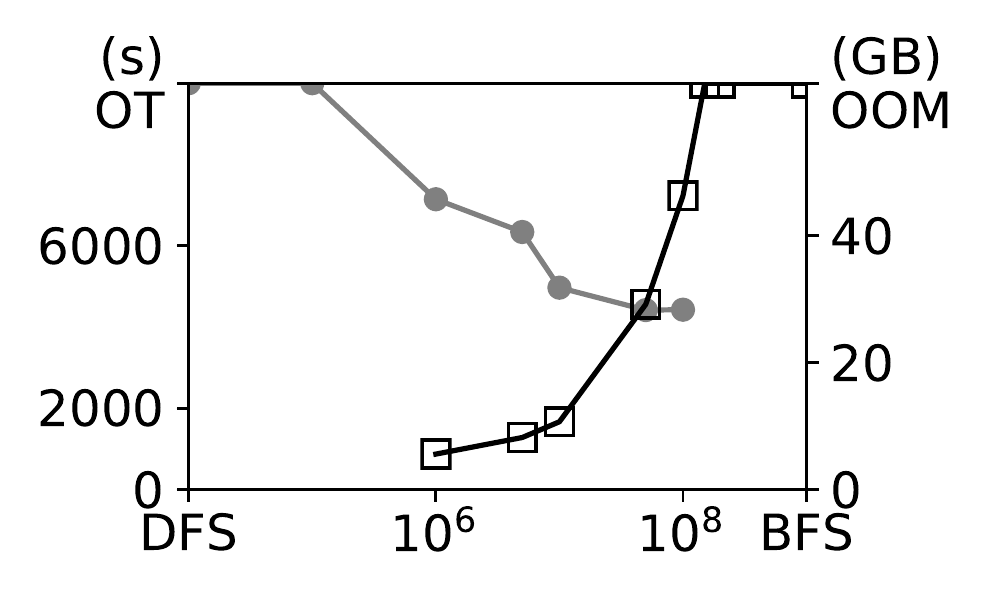}  
      \caption{Dynamic Scheduling}
      \label{fig:exp3-scheduling}
\end{figure}

\stitle{Exp-7: Scheduling.} We evaluate the BFS/DFS-adaptive scheduling using a long-running query $q_6$ that can trigger memory crisis.
By varying the output queue size for each operator from $0$ to infinite, \page's scheduler essentially transforms from DFS scheduler, to BFS/DFS-adaptive scheduler, and eventually to BFS scheduler. \reffig{exp3-scheduling} shows the execution time and memory consumption for different queue sizes. When the queue size is smaller than $10^6$ (including pure DFS scheduling), the query runs \timeout. As the size increases, the execution time decreases until $10^7$, from which the curve gets flat. The execution time at the point $5{\times}10^7$ is 38\% faster than that at $1{\times}10^6$. After the size goes beyond $10^8$ (including BFS-style scheduling), the program encounters \oom. The results indicate that \page's adaptive scheduler keeps the memory usage bounded while achieving high efficiency. 

\begin{figure}[h]
    \centering
    \includegraphics[width=0.8\linewidth]{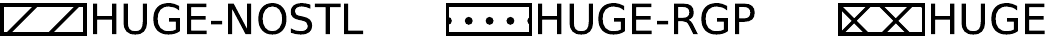}
    \\
    \begin{subfigure}{.25\textwidth}
      \centering
      \includegraphics[width=\linewidth]{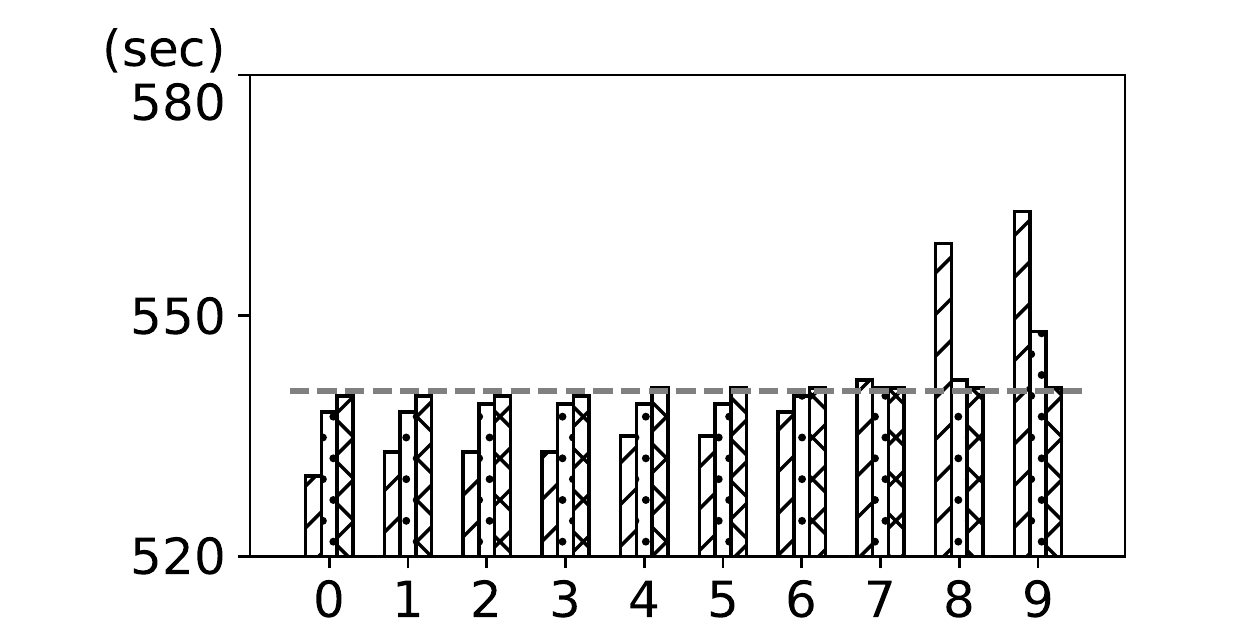}  
      \caption{$q_1$}
    \end{subfigure}
     ~
    \begin{subfigure}{.25\textwidth}
      \centering
      \includegraphics[width=\linewidth]{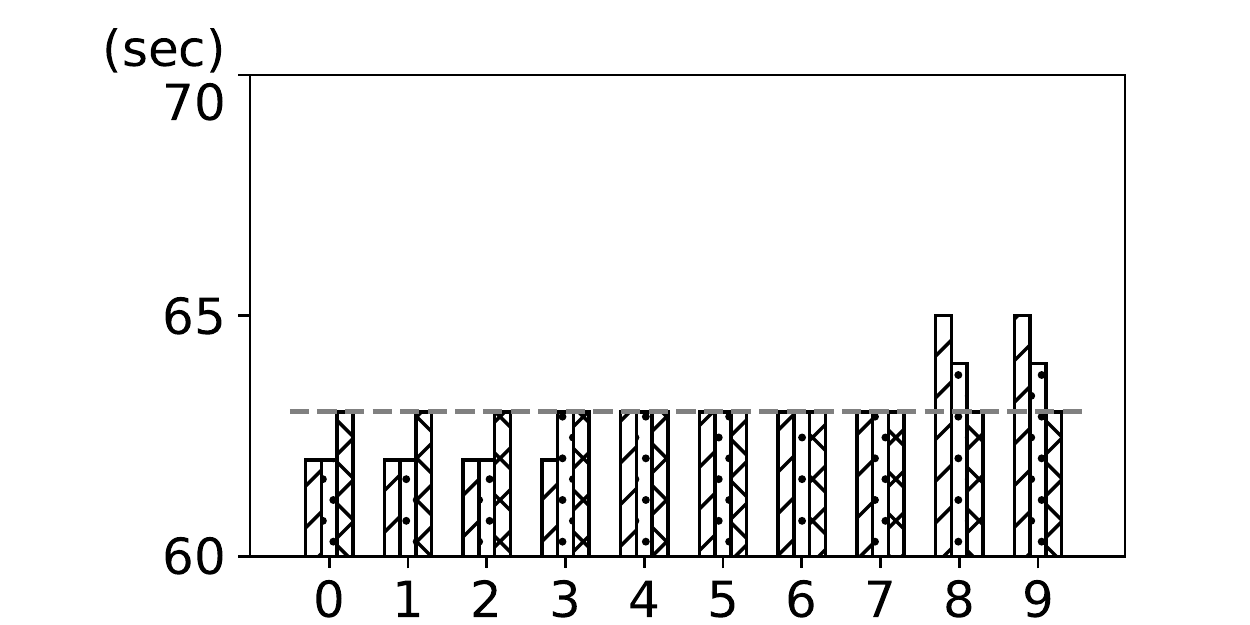} 
      \caption{$q_2$}
    \end{subfigure}
    \\
        \begin{subfigure}{.25\textwidth}
      \centering
      \includegraphics[width=\linewidth]{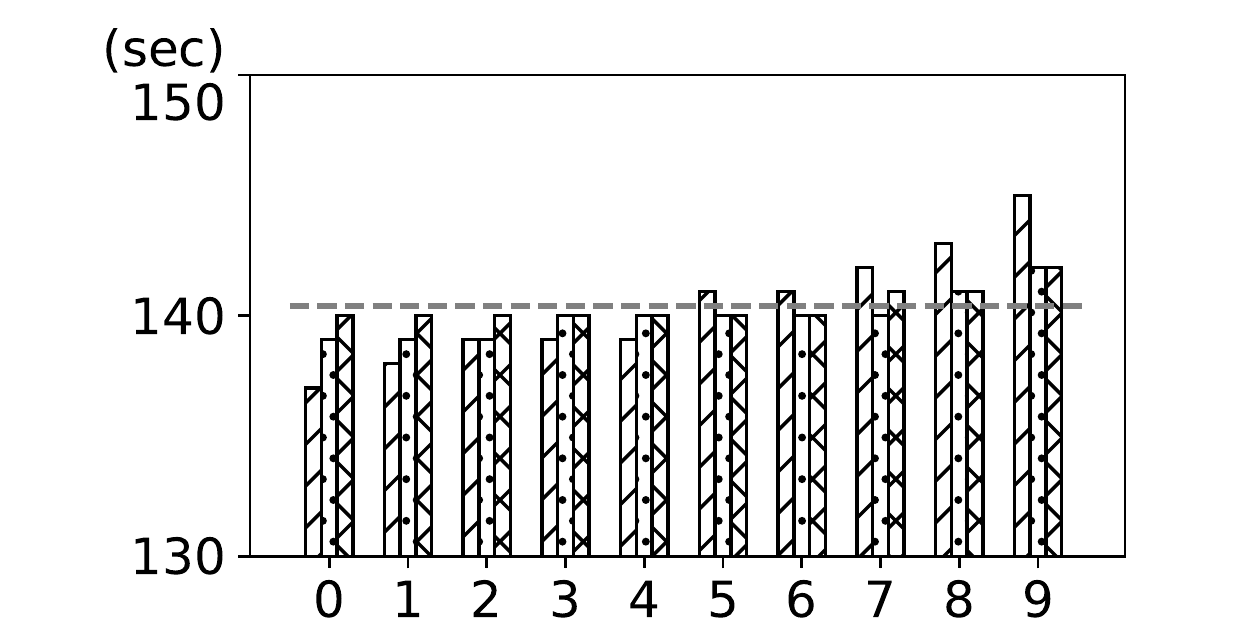}  
      \caption{$q_3$}
    \end{subfigure}
     ~
    \begin{subfigure}{.25\textwidth}
      \centering
      \includegraphics[width=\linewidth]{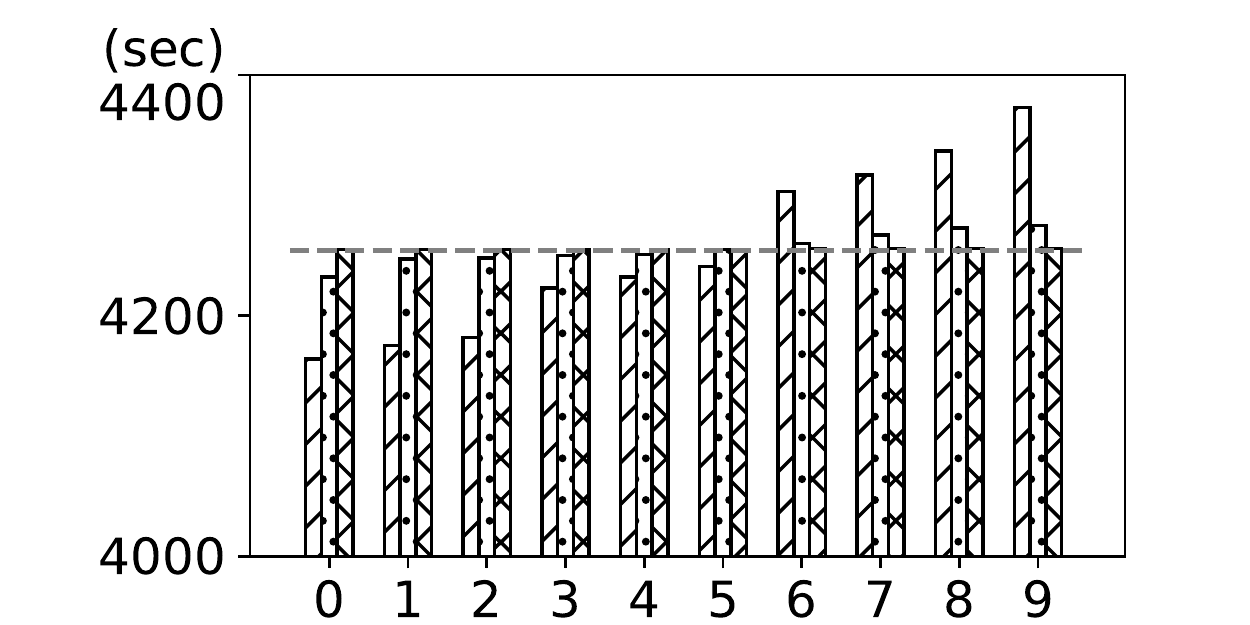} 
      \caption{$q_6$}
    \end{subfigure}

    \caption{Work Stealing}
    \label{fig:work-stealing}
\end{figure}

\stitle{Exp-8: Load balancing.} To test the effectiveness of our work-stealing technique, we compare \page with \page-NOSTL (\page with work stealing disabled, which distributes the load on the pivot vertex as \benu) and \page-RGP (with the region-group technique of \rads).
The results are shown in \reffig{work-stealing} (other queries run \timeout).  
We measure the standard deviation of the execution time among all workers. 
Take $q_6$ as an example, 
with the help of work stealing, \page demonstrates the best load balancing with a standard deviation of only $0.5$, compared to \page-NOSTL's $73.4$ and \page-RGP's $13.2$, which can also be observed from \reffig{work-stealing}. We then measure the overhead of our work-stealing technique by computing $Total$ as the aggregated CPU time among all workers. Compared to \page-NOSTL, \page only slightly increases $Total$ by $0.017\%$.


\begin{table}[h]
\caption{Runtime for Different Execution Plans}
\label{tab:execution-plans}
\small
\setstretch{0.9}
\begin{tabular}{|c|c|c|c|c|}
\hline
 & \pagewco  & \pageeh & \pagegf & \page    \\ \hline\hline
$q_7$ & \timeout  & \multicolumn{3}{c|}{\textbf{7340.28s (170.02s)}}    \\ \hline
$q_8$ & $64.5s (21ms)$ & 67.2s (15.6s)  & 64.4s (13.9s) & \textbf{40.1s (6.5s)}    \\ \hline
\end{tabular}
\end{table}

\stitle{Exp-9: Comparing Hybrid Plans.} We plug into \page the logical plans of \wopt join (as \pagewco), as well as the hybrid plans of \eh (\pageeh) and \gf (\pagegf), and compare them with \page (with the plan by \refalg{dp-opt}). We use queries $q_7$ and $q_8$ for their variances in execution plans, and the graph \dtgo to avoid too many \timeout cases.
For $q_7$, the optimiser of \page produces the same logical plan as \eh and \gf that joins a 3-path with a 2-path (via \join operator), which is better than the \wopt join plan that must produce the matches of a 4-path. For $q_8$, \page's optimiser, \eh and \gf all generate their own hybrid plans, while \page's plan renders better performance. This is because that \page's optimiser takes both computation and communication into consideration, while existing hybrid plans are developed in the sequential context where computation is the only concern (\refex{optimal-plan}).

\begin{figure}[h]
    \centering
    \begin{subfigure}[b]{0.5\textwidth}
          \centering
          \includegraphics[height = 0.08in]{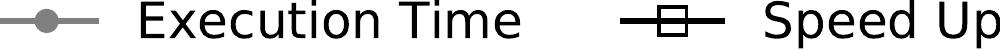}
    \end{subfigure}%
    \\
    \begin{subfigure}{.2\textwidth}
          \centering
          \includegraphics[width=\linewidth]{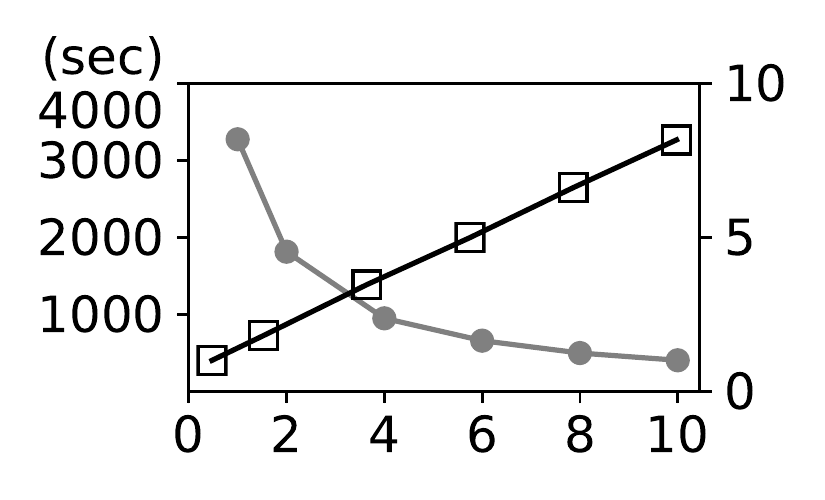}  
          \vspace{-2\baselineskip}
          \caption{\page $q_2$}
    \end{subfigure} 
    ~
    \begin{subfigure}{.2\textwidth}
          \centering
          \includegraphics[width=\linewidth]{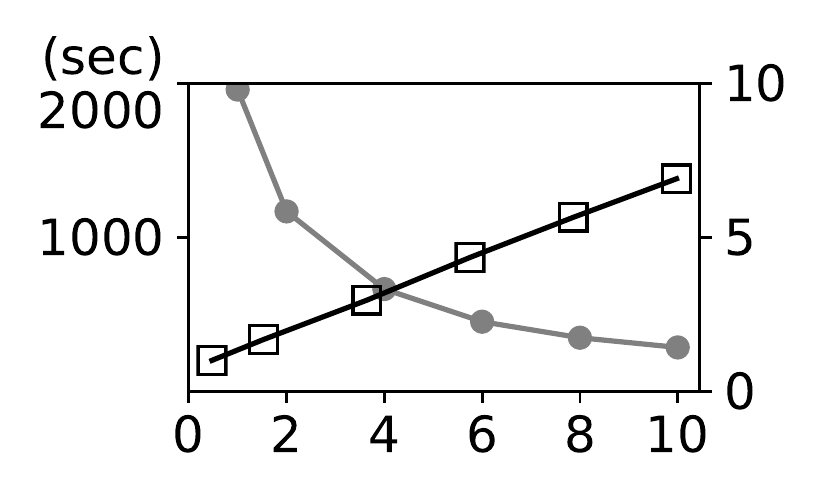}  
          \vspace{-2\baselineskip}
          \caption{\page $q_3$}
    \end{subfigure}
    \\
    \begin{subfigure}{.2\textwidth}
          \centering
          \includegraphics[width=\linewidth]{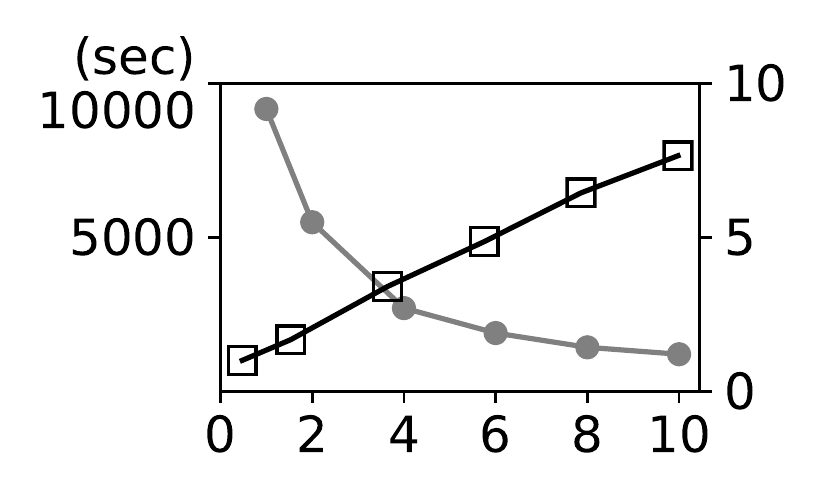} 
          \vspace{-2\baselineskip}
          \caption{\bigjoin $q_2$}
    \end{subfigure} 
    ~
    \begin{subfigure}{.2\textwidth}
          \centering
          \includegraphics[width=\linewidth]{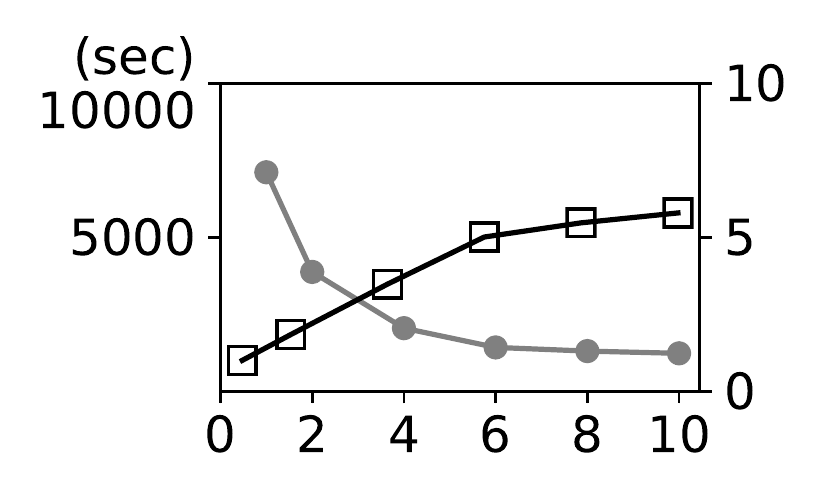}  
          \vspace{-2\baselineskip}
          \caption{\bigjoin $q_3$}
    \end{subfigure}

    \caption{Vary Number of Machines.}
    \label{fig:exp5-scalability}
\end{figure}

\stitle{Exp-10: Scalability.} We test the scalability of \page by varying the machine number in the cluster from 1 to 10 using the bigger data graph \dtfs (\reffig{exp5-scalability}). \page demonstrates almost linear scalability for both queries. 
Additionally, we compare the scalability of \page with \bigjoin (the \timeout results of \benu, \rads and \seed are excluded). \page scales better than \bigjoin, with the average scaling factor of $7.5\times$ compared to \bigjoin's $6.7\times$ from 1 to 10 machines. 

\section{Related Work}
\label{sec:related}

\stitle{Single-machine Enumeration} Many sequential algorithms are proposed, mainly for labelled graphs. Most of them follows Ullmann's \cite{ullmann} backtracking approach \cite{comparison,match-survey} with different matching order, pruning rules, and index structure 
\cite{vf2,vf3,quicksi,graph-ql,GADDI,spath,turbo-iso,cfl,DAF,dualsim}. 
Parallel algorithms 
\cite{PGX.ISO,RI,PSM,LIGHT,CECI,jin2021fast} 
are later proposed to compute subgraph matching using multiple threads. Similarly, \eh \cite{empty-headed} and \gf \cite{graphflow-demo,graphflow} aim at computing subgraph enumeration in parallel on a single machine mixing worst-case optimal join \cite{Ngo-join} and binary join.  
They can be seamlessly migrated to distributed environment using \page (\refsec{plan}).


\stitle{Distributed Enumeration} \multiwayjoin \cite{multiway-join} uses a one-round multiway-join to enumerate subgraphs, and QFrag \cite{Qfrag} broadcasts the data graph,
These algorithms have poor scalability for large data graphs or complex queries \cite{patmat-exp}. Another class of algorithms, including \edgejoin \cite{edge-join}, \starjoin \cite{star-join}, \psgl \cite{psgl}, \ttjoin \cite{twin-twig}, \seed \cite{seed}, \cbf \cite{crystaljoin}, and \bigjoin \cite{wco-join}, follows a join-based framework that has been empirically studies and evaluated in \cite{patmat-exp}. 
To solve the problem of shuffling huge amount of intermediate results in join-based algorithms, \cite{crystaljoin} proposed a compression technique to reduce communication cost. \benu and \rads further introduced a pull-based scheme that pull the data graph when needed instead of shuffling the intermediate results. However, they do not demonstrate satisfactory performance as illustrated in this paper.
\section{Conclusion}
\label{sec:conclusion}
In this paper, we present \page, an efficient and scalable subgraph enumeration system in the distributed context. \page incorporates an optimiser to compute an advanced execution plan, and a novel architecture that supports pulling/pushing-hybrid communication. Together with a lock-free and zero-copy cache design, and a dynamic scheduling module, \page achieves high performance in both computation and communication with bounded memory. 



\balance

\bibliographystyle{ACM-Reference-Format}
\bibliography{sample-base}


\end{document}